\documentclass[preprint]{elsarticle}
\pdfoutput=1
\usepackage{aliases}

\numberwithin{equation}{section} \allowdisplaybreaks

\begin{document}

\title{A discontinuous Galerkin method for a diffuse-interface model of immiscible two-phase flows with soluble surfactant}
%the advective Cahn-Hilliard equations with surfactant}

\author[1]{Deep Ray}
\author[2]{Chen Liu}
\author[3]{Beatrice Riviere}
\address[1]{Department of Aerospace and Mechanical Engineering\\University of Southern California, Los Angeles, TX, USA \\ deepray@usc.edu}
\address[2]{Department of Computational and Applied Mathematics \\Rice University, Houston, TX, USA \\  cliu.chemaths@gmail.com}
\address[3]{Department of Computational and Applied Mathematics \\Rice University, Houston, TX, USA \\  riviere@rice.edu}

\date{}

\begin{abstract}
A numerical method using discontinuous polynomial approximations is formulated 
for solving a phase-field model of two immiscible fluids with a soluble surfactant.
The scheme recovers the Langmuir adsorption isotherms at equilibrium.  Simulations
of spinodal decomposition, flow through a cylinder and flow through a sequence of pore throats
show the dynamics of the flow with and without surfactant. Finally the numerical method is used
to simulate fluid flows in the pore space of Berea sandstone obtained by micro-CT imaging.
\end{abstract}
%\keywords{Cahn-Hilliard equation \and Soluble surfactant \and Diffuse-interface \and Discontinuous Galerkin \and Porous media \and Adsorption isotherm}

\maketitle

%  % !  TEX root = adv_surch_scheme.tex
%\documentclass[adv_surch_scheme.tex]{subfiles}
%
%\begin{document}
\section{Introduction}
Surfactants, or surface active agents, play a crucial rule in various industrial and biochemical processes. These include the use of detergent to remove greasy stains \cite{NEUG90}, emulsification agents used to increase the shelf life of food \cite{HASEN19}, surfactant-flooding for efficient recovery of oil from reservoirs \cite{CHEN17} and pulmonary surfactants that prevent lung collapse \cite{HALPERN98}. Surfactant molecules adhere to the interface of two phases (liquid-liquid, liquid-gas or liquid-solid) and lower the interfacial surface tension, thereby increasing the miscibility of the two components. Since surfactants can significantly alter the dynamics of binary mixtures, it becomes necessary develop suitable mathematical models to capture their interaction.

There are primarily two family of methods used to model interfacial dynamics of a multiphase system in the presence of a surfactant. The first corresponds to the sharp-interface methods (see \cite{LOWEN11} and references within), where the interface is considered to be infinitesimally thin. The interface can be tracked explicitly using boundary integral methods \cite{STONE90,MILL93,POZ97}, front-tracking methods \cite{ZHANG2006,MURA2008,LAI2008}, or implicitly via level-sets \cite{XU2003}, volume of fluid \cite{RENARDY2002,LOWEN04} or arbitrary Langrangain-Eulerian methods \cite{HAMEED2008}.  A suitable partial differential equation is formulated to describe the evolution of the surfactant at the interface. In order to simulate soluble surfactants and enable mass transfer across the interface, external source terms and boundary conditions need to be introduced, which need not arise naturally from the model itself.

The second class of methods are the diffusive-interface models based on thermodynamics and density gradient theory \cite{CAHN58}. The interface is considered to have a width which describes the zone of phase-transition and which typically scales as the measure of spatial discretization.  These methods require the specification of a suitable free-energy functional which captures the key dynamics in the bulk phase and the interface. A big appeal of diffusive-interface methods over the sharp-interface methods is that the entire system of equations describing the evolution of the various mixture components and other quantities of interest can be derived from a single energy functional, thus leading to a consistent thermodynamical model formulation. Several free energy formulations have been proposed \cite{LOWEN11,LARADJI92,PATZ95,DIAM96,KOMURA97,DIAM01,SMAN06,LIU10,ENGBLOM13,YANG18,YANG2017,ZHU18}, each having their own advantage. A few of these choices are motivated by the faithful recovery adsorption isotherms \cite{SMAN06,DIAM96,LIU10,ENGBLOM13}. There has also been an active interest in developing energy stable numerical methods which ensure the consistent decay of total energy \cite{YANG18,YANG2017,ZHU18}.  

In a recent series of works \cite{FRANK18,FRANK18b,LIU20}, a diffusive-interface framework was considered for an immiscible two-phase flows at the pore-scale in rock samples. The location of the two-phases in the pore space of the rock is expressed in terms of an order parameter, which may be defined as the difference between mass fractions. Capillary forces and viscous forces drive the displacement of the two phases through the network of connected pores and pore throats. The system is mathematically modeled by the Cahn-Hilliard equations coupled with the incompressible Navier-Stokes equations.  
An interior penalty discontinuous Galerkin (IPDG) scheme was proposed to solve the system, while a temporal semi-implicit convex-concave splitting ensured the scheme to be unconditionally energy stable 
\cite{FRANK18}. The coupled Cahn-Hilliard-Navier-Stokes problem has received much attention recently and several numerical methods have been employed to solve this problem, namely finite element methods and mixed element methods in \cite{feng2006fully,BaoShiSunWang2012,GuoLinLowengrub2014}, finite volume methods \cite{KouSunWang2018} and discontinuous Galerkin methods \cite{GiesselmannPryer2015,LIU20}.

In the present work, we consider a system with three-components: two components form two immiscible phases and the third component
is the surfactant that is miscible in both phases.  The mathematical model is based on the free-energy functional proposed
in \cite{ENGBLOM13}, which leads to equations that are more complex than the Cahn-Hilliard equations for a two-phase system.  
The three-component system is advected by a given velocity field that has been obtained by solving the incompressible Navier-Stokes equations in the pore space.
The primary objectives of this work are:
\begin{itemize}
\item Construct an IPDG scheme for the advective three-component system that is energy dissipative. 
\item Demonstrate the capability of the scheme to recover adsorption isotherms, while emulating key surfactant dynamics.
\item Effectively simulate the flow in porous structures, including a digital rock obtained by 3D imaging of micro-CT slices of the real rock samples.
\end{itemize}

The rest of the paper is organized as follows. Section~\ref{sec:model} describes the mathematical model and formulates the non-dimensional system of partial differential equations describing the flow. In Section~\ref{sec:scheme}, the spatial and temporal discretization is discussed, along with a proof for the decay of total energy at the discrete level. Several numerical results are presented in Section~\ref{sec:results} to demonstrate the performance of the scheme, followed by concluding remarks in the last section.

%\end{document}

%  % !  TEX root = main.tex
%\documentclass[main.tex]{subfiles}

%\begin{document}
\section{Mathematical model}
\label{sec:model}
A number of models are available in literature \cite{LARADJI92,KOMURA97,LOWEN04,SMAN06,LIU10,LOWEN11,ENGBLOM13,ZHU18} to describe the propagation of an incompressible binary mixture in the presence of a surfactant. Each model is endowed with its own set of advantages in capturing realistic flow behaviour and ensuring stable numerical computations.  In this work, we choose the diffuse-interface model proposed in \cite{ENGBLOM13} to balance the model complexity while ensuring a faithful representation of the underlying physics.  

\subsection{Governing equations}
Let $\Omega \subset \Ro^3$ be an open bounded polyhedral domain
and let $(0,T)$ denote the time interval with $T\in \Ro^+$. We use the notation $\Omega_T := \Omega \times (0,T)$ to donate the combined space-time domain. We denote by $c: \Omega_T  \mapsto [-1,1]$ the order parameter, which is  the difference between  mass (or volume) fractions of the two components of the mixture. Let us denote the surfactant volume fraction by 
$s : \Omega_T \mapsto [0,1]$.  The Helmholtz free energy of the system (see \cite{ENGBLOM13} and references therein) can be expressed as

\begin{subequations}\label{eqn:helmenergy}
\begin{align}
\mathcal{F}(c,s) &= \int_\Omega \left( F_c + F_s + F_{s,c}\right),\tag{\ref{eqn:helmenergy}}\\
F_c &= \beta_1 \Phi(c) + \frac{\kappa}{2} |\nabla c|^2, \\
F_s &= \beta_2 \Psi(s), \\
F_{s,c} &= - \beta_3 s \Phi(c) + \beta_4 s c^2,  
\end{align}
\end{subequations}
where $\kappa, \beta_1,\beta_2,\beta_3,\beta_4$ are non-negative constants.
In the above equations, $F_c$ is the energy functional for the two immiscible bulk phases, $F_s$ is the energy associated with the local surfactant concentration, and $F_{s,c}$ is the contribution to the energy from the interaction between the surfactant and the two phases. The term $(-s\Phi(c))$ is the energy potential accounting for the adsorption of the surfactant at the interfacial boundary, and 
the term $s c^2$ penalizes the amount of free surfactant in the bulk phases. For the remainder of this paper, we choose $\Phi(c)$ to be the Ginzburg-Landau double well potential
\begin{equation}\label{eqn:Psi}
\Phi(c) = \frac{1}{4}(1-c^2)^2,
\end{equation}
and $\Psi$ to be the entropic part of the Flory-Huggins potential
\begin{equation}\label{eqn:Psi}
\Psi(s) = s \log(s) + (1-s) \log(1-s) + \log(2),
\end{equation}
where the last constant term is added to ensure $\Psi$ is non-negative. Since $\Psi$ is ill-defined as $s$ approaches 0 or 1, we implement the following regularized version of the potential
\begin{equation}\label{eqn:Psi_reg}
\Psi(s) = \begin{cases}
s \log(s) + (1-s) \log(1-s)  \\
\ \ \ + \log(2)  \qquad \ \ \ \ \ \text{ if } s \in [\epsilon, 1 - \epsilon ],\\
s\log(s) + \frac{1}{2\epsilon}(1-s)^2 + (1-s)\log(\epsilon) \\
\ \ \ - \frac{\epsilon}{2} + \log(2) \qquad \text{ if } s > 1 - \epsilon,\\
(1-s)\log(1-s) + \frac{1}{2\epsilon}s^2 + s\log(\epsilon) \\
\ \ \ - \frac{\epsilon}{2} + \log(2) \qquad \text{ if } s < \epsilon,\\
\end{cases}
\end{equation}
with the threshold $\epsilon = 10 ^{-6}$.

The potential $\Phi(s)$ can be decomposed into the sum of a convex part $\Phi_+$ and a concave part $\Phi_-$. Although this splitting is not unique, we make the following choice in this paper
\begin{equation}\label{eqn:splitting}
\Phi_+(c) = \frac{1}{4}(1 + c^4), \quad \Phi_-(c) = - \frac{1}{2} c^2.
\end{equation}
Furthermore, $\Psi(s)$ is a convex function whenever $s \in [0,1]$.

\begin{remark}
Three free-energy models were considered in \cite{ENGBLOM13}. The choice \eqref{eqn:helmenergy} corresponds to "Model 3" with suitably chosen values for $\beta_i$. 
\end{remark}

Taking the functional/variational derivative of the Helmholtz energy with respect to $c$ and $s$ leads to the following expressions of the chemical potentials
\begin{equation}\label{eqn:chem_pot}
\begin{aligned}
\cpc&:= \frac{\delta \mathcal{F}}{\delta c} = \beta_1 \Phi^\prime(c) - \kappa \Delta c - \beta_3 s \Phi^\prime(c) + 2 \beta_4 c s, \\ \cps&:=\frac{\delta \mathcal{F}}{\delta s} = \beta_2 \Psi^\prime(s) - \beta_3 \Phi(c) + \beta_4 c^2.
\end{aligned}
\end{equation}
Let $\vel$ be a solenoidal velocity field. The order parameter and surfactant satisfy the mass balance equations:
\begin{align*}
\partial_t c - \nabla \cdot (M_c \nabla \cpc) + \nabla \cdot (c \vel) &=0 \quad & \text{in } \Omega_T ,\\
\partial_t s - \nabla \cdot (M_s \nabla \cps) + \nabla \cdot (s \vel)&=0 \quad & \text{in } \Omega_T,
\end{align*}
where $M_c$ and $M_s$ are non-negative mobilities. In order to remove the dependence of the surfactant Cahn-Hilliard model on physical units, we appropriately non-dimensionalize the equations.  We begin by listing the main quantities and their units in Table \ref{tab:units}. Let us denote the characteristic length as $\bar{x}$, the characteristic velocity as $\bar{v}$, the characteristic time as $\bar{t} = \bar{x}/\bar{v}$, the characteristic chemical potential as $\bar{\mu} = \beta_1$, the characteristic mobility (for $c$) as $\bar{M_c}$ and the characteristic mobility (for $s$) as $\bar{M_s}$. 
\begin{table}[!h]\label{tab:units}
\begin{center}
\begin{tabular}{ccc}
\hline
 \textbf{Quantity} & \textbf{Symbol} & \textbf{Unit} \\ \hline
 time &  $t$ & s \\
 length & $x$ & m \\
 order parameter&  $c$ & - \\
 surfactant &  $s$ & - \\  
 chemical potential  & $\cpc$,  $\cps$ & kg m$^{-1}$ s$^{-2}$\\
 mobility& $M_c$, $M_s$ & m$^3$ s kg$^{-1}$\\
coefficient (type 1)&  $\kappa$ &  kg m s$^{-2}$ \\
coefficient (type 2) & $\beta_1, \beta_2, \beta_3, \beta_4$ &  kg m$^{-1}$ s$^{-2}$\\ \hline
\end{tabular}
\caption{Quantities of model \eqref{eqn:model} and their units.}
\end{center}
\end{table}
We define the Peclet $(\Pe_c,\Pe_s)$  and Cahn $(\Cn)$ numbers:
\[
\Pe_c = \frac{\bar{x}^2}{\beta_1 \bar{t} \bar{M}_c}, \quad
\Pe_s = \frac{\bar{x}^2}{\beta_1 \bar{t} \bar{M}_s}, \quad 
\Cn = \left(\frac{\kappa}{\bar{x}^2 \beta_1}\right)^{1/2}.
\]
The non-dimensional equations are (for simplicity, we keep the same notation for the dimensionless quantities):
\begin{subequations}\label{eqn:model}
\begin{align}
\partial_t c - \frac{1}{\Pe_c}\nabla \cdot (M_c \nabla \cpc) + \nabla \cdot (c \vel) &=0 \quad & \text{in } \Omega_T,\\
\partial_t s - \frac{1}{\Pe_s}\nabla \cdot (M_s \nabla \cps) + \nabla \cdot (s \vel)&=0 \quad & \text{in }\Omega_T,\\
\cpc - \Phi^\prime(c) + \Cn^2 \Delta c && \notag \\ 
+ \alpha_3 s \Phi^\prime(c) - 2 \alpha_4 c s &=0 \quad & \text{in } \Omega_T,\\
\cps - \alpha_2 \Psi^\prime(s) + \alpha_3 \Phi(c) - \alpha_4 c^2 &=0 \quad & \text{in } \Omega_T,
\end{align}
\end{subequations}
where $M_c$ is a dimensionless constant, $M_s$ is taken to be the function $M_s = \max\bigl(0,s(1-s)\bigr)$ and the remaining non-dimensional coefficients are:
\[
\alpha_i = \frac{\beta_i}{\beta_1}, \quad 2\leq i \leq 4.
\]
The initial conditions for the system \eqref{eqn:model} are given by $c^0:\bar{\Omega}\mapsto [-1,1]$ and $s^0:\bar{\Omega}\mapsto [0,1]$. In order to prescribe boundary conditions, let us partition the domain boundary $\dbnd$. We use the notation $\Gamma^\mathrm{wall}$ to denote the part of the domain boundary that corresponds to the fluid-solid interface, where a no-slip boundary condition is assumed for the velocity field, i.e., $\vel = 0$. If $\dbnd = \Gamma^\mathrm{wall}$, then the system is said to be closed. In addition to this, $\dbnd$ may be further partitioned into the inflow and outflow boundaries
\begin{equation*}
\begin{aligned}
\Gamma^\mathrm{in} &= \{ \x \in \dbnd \ : \ \vel \cdot \n < 0 \}, \\
\Gamma^\mathrm{out} &= \partial\Omega \setminus (\Gamma^\mathrm{wall} \cup \Gamma^\mathrm{in}),
\end{aligned}
\end{equation*}
where $\n$ denotes the unit normal vector outward of the domain. 
We consider the following boundary conditions
\begin{subequations}\label{eqn:surch_bc}
\begin{align}
c &= c_\mathrm{in}, \quad \text{on } \Gamma^\mathrm{in} \times (0,T),  \label{eqn:surch_bca} \\
s &= s_\mathrm{in}, \quad \text{on } \Gamma^\mathrm{in} \times (0,T),  \label{eqn:surch_bcb} \\
\nabla c \cdot \n&=0  \quad\text{on } (\Gamma^\mathrm{wall} \cup \Gamma^\mathrm{out}) \times (0,T),  \label{eqn:surch_bcc} \\
M_c \nabla \cpc \cdot \n&=0  \quad\text{on } \dbnd \times (0,T), \label{eqn:surch_bcd} \\
M_s \nabla \cps \cdot \n&=0  \quad\text{on } \dbnd \times (0,T), \label{eqn:surch_bce}
\end{align}
\end{subequations} 
where $c_\mathrm{in}:\Gamma^\mathrm{in} \times (0,T) \mapsto [-1,1]$ 
and $s_\mathrm{in}:\Gamma^\mathrm{in} \times (0,T) \mapsto [0,1]$.

\subsection{Energy decay and mass conservation}
Assuming that $\vel = 0$ in $\overline{\Omega_T}$, i.e., the system is non-advective, the total Helmholtz energy \eqref{eqn:helmenergy} is non-increasing in time. Indeed, using the system \eqref{eqn:model} with the boundary conditions \eqref{eqn:surch_bc}, we obtain
% TWO COLM FORMAT
%\clrg{\begin{align*}
%\dd{\mathcal{F}}{t} =& \int_\Omega \frac{\delta \mathcal{F}}{\delta c} \partial_t c + \int_\Omega \frac{\delta \mathcal{F}}{\delta s} \partial_t s \\
%=& - \frac{1}{\Pe_c}\int_\Omega M_c | \nabla \cpc|^2 - \frac{1}{\Pe_s} \int_\Omega M_s | \nabla \cps|^2  \\
%\leq& \ 0.
%\end{align*}}
% SINGLE COLM FORMAT
\begin{align*}
\dd{\mathcal{F}}{t} =& \int_\Omega \frac{\delta \mathcal{F}}{\delta c} \partial_t c + \int_\Omega \frac{\delta \mathcal{F}}{\delta s} \partial_t s 
= - \frac{1}{\Pe_c}\int_\Omega M_c | \nabla \cpc|^2 - \frac{1}{\Pe_s} \int_\Omega M_s | \nabla \cps|^2  
\leq \ 0.
\end{align*}
For non-advective closed systems, we can easily show that
\[
\int_\Omega c = \int_\Omega c_0, \quad \int_\Omega s = \int_\Omega s_0
\]
This implies that the mass of the surfactant is conserved.  This also implies that the mass of the two components that form
the two immiscible phases is conserved.

%In this paper, we propose a numerical scheme that ensures the decay of the Helmholtz energy at the discrete level.
%\end{document}

%  % !  TEX root = adv_surch_scheme.tex
%\documentclass[adv_surch_scheme.tex]{subfiles}
%
%\begin{document}

\section{Discretization}
\label{sec:scheme}
In this section, we give details of the discrete spaces and operators needed to formulate the discontinuous Galerkin (DG) scheme for \eqref{eqn:model}. We first describe the temporal discretization by assuming continuity in space.

\subsection{Temporal discretization}\label{sec:temporal_disc}
Let $0= t_0 < t_1 < ... < t_{N_T}$ be a decomposition of $(0,T)$ into $N_T$ subintervals, with $\tau_n = t_n - t_{n-1}$ denoting the $n$th step size. The velocity field $\vel$ is given at each time step $t_n$ and it is denoted by $\vel^n$. Then the semi-discrete (in time) scheme reads as follows:

For each $1 \leq n \leq N_T$, given $(c^{n-1}, s^{n-1})$ find $c^n$, $s^n$, $\mu_c^n$, $\mu_s^n$ such that
% TWO COLM FORMAT
%\begin{subequations}\label{eqn:timealg}
%\begin{align}
%\frac{c^n-c^{n-1}}{\tau_n} - \frac{1}{\Pe_c} \nabla \cdot (M_c \nabla \mu_c^n) &&& \notag\\
%+ \nabla \cdot (c^n \vel^n) & = 0  && \text{in } \Omega, \label{eqn:timealga}\\
%\frac{s^n-s^{n-1}}{\tau_n} - \frac{1}{\Pe_s} \nabla \cdot (M_s(s^{n-1}) \nabla \mu_s^n) &&& \notag\\
%+ \nabla \cdot (s^n \vel^n) & = 0  && \text{in } \Omega, \label{eqn:timealgb}\\
%-\mu_c^n +   \Phi_+^\prime(c^n) + \Phi_-^\prime(c^{n-1}) - \Cn^2 \Delta c^n&&& \notag\\
% -\alpha_{3} s^n (\Phi_+^\prime(c^{n-1}) +\Phi_-^\prime(c^{n})) &&& \notag\\
%+ \alpha_{4} s^n (c^{n-1} + c^{n}) &= 0 && \text{in } \Omega ,\label{eqn:timealgc}\\
%-\mu_s^n  + \alpha_2 \Psi^\prime(s^n) - \alpha_{3} \Phi(c^{n-1})  &&& \notag\\
%+\alpha_{4} (c^{n-1})^2 &= 0  && \text{in } \Omega,\label{eqn:timealgd} 
%\end{align}
% SINGLE COLM FORMAT
\begin{subequations}\label{eqn:timealg}
\begin{align}
\frac{c^n-c^{n-1}}{\tau_n} - \frac{1}{\Pe_c} \nabla \cdot (M_c \nabla \mu_c^n) + \nabla \cdot (c^n \vel^n) & = 0  \qquad \text{in } \Omega, \label{eqn:timealga}\\
\frac{s^n-s^{n-1}}{\tau_n} - \frac{1}{\Pe_s} \nabla \cdot (M_s(s^{n-1}) \nabla \mu_s^n) + \nabla \cdot (s^n \vel^n) & = 0  \qquad \text{in } \Omega, \label{eqn:timealgb}\\
-\mu_c^n +   \Phi_+^\prime(c^n) + \Phi_-^\prime(c^{n-1}) - \Cn^2 \Delta c^n -\alpha_{3} s^n (\Phi_+^\prime(c^{n-1}) +\Phi_-^\prime(c^{n})) & \notag\\
+ \alpha_{4} s^n (c^{n-1} + c^{n}) &= 0 \qquad \text{in } \Omega ,\label{eqn:timealgc}\\
-\mu_s^n  + \alpha_2 \Psi^\prime(s^n) - \alpha_{3} \Phi(c^{n-1})  +\alpha_{4} (c^{n-1})^2 &= 0  \qquad \text{in } \Omega,\label{eqn:timealgd} 
\end{align}
with the initial and boundary conditions
\begin{align}
c^n & = c_\mathrm{in}   \qquad \text{on}\ \Gamma^{\mathrm{in}}, \\
s^n & = s_\mathrm{in}   \qquad \text{on}\ \Gamma^{\mathrm{in}}, \\
\nabla c^n \cdot \n&=0  \ \ \qquad \text{on } \Gamma^\mathrm{wall} \cup \Gamma^\mathrm{out},  \label{eqn:step1e} \\
M_c \nabla \mu_c^n \cdot \n&=0  \ \ \qquad \text{on } \partial\Omega, \label{eqn:step1f} \\
M_s(s^{n-1}) \nabla \mu_s^n \cdot \n&=0  \ \ \qquad \text{on } \partial\Omega. \label{eqn:step1g}
\end{align}
\end{subequations}

The semi-implicit time discretization considered above is useful in constructing an energy decaying scheme (see Section \ref{sec:disc_energy_decay}).

\subsection{Spatial discretization}
The spatial domain $\Omega$ is discretized using a family of conforming non-degenerate regular meshes $\Tau_h = \{E_k\}$, where $h$ denotes the maximum element diameter. We denote by $\Gamma_h$ the set of interior faces.
For each $e \in \Gamma_h$ shared by elements $E_{k^-}$ and $E_{k^+}$, we define the unit normal vector $\n_e$ oriented from $E_{k^-}$ to $E_{k^+}$ if $k^- < k^+$. Note that for $e \in \dbnd$, $\n_e$ denotes the outward unit normal to $\dbnd$. The average and jump of any scalar quantity $w$ across the face $e$ is denoted by
\begin{equation*}
\begin{aligned}
\eavg{w} &= \begin{cases} \frac{1}{2} w \big|_{E_{k^-}}  + \frac{1}{2} w \big|_{E_{k^+}}  \quad & \text{if} \ e \in \Gamma_h\\
                                              w \big|_{E_{k^-}} \quad & \text{if} \ e \in E_{k^-} \cap \dbnd
                     \end{cases},                         \\
\edif{w} &= \begin{cases}  w \big|_{E_{k^-}}  - w \big|_{E_{k^+}} \quad & \text{if} \ e \in \Gamma_h\\
                                              w \big|_{E_{k^-}} \quad & \text{if} \ e \in E_{k^-} \cap \dbnd
                     \end{cases} .                     
\end{aligned}
\end{equation*}

For any positive integer $r$, consider the broken Sobolev space
\[
H^r(\Tau_h) = \{ v \in L^2(\Omega):\ \forall E \in \Tau_h, v \big|_{E} \in H^r(E) \}. 
\]
We define the following discrete forms for the various differential operators in \eqref{eqn:model}
\begin{align*}
a_\mathcal{A} &: H^2(\Tau_h)^d \times H^2(\Tau_h)\times H^2(\Tau_h) \mapsto \Ro, \\
b_\mathcal{A} &:  H^2(\Tau_h)^d \times L^\infty(\Tau_h) \times H^2(\Tau_h)  \mapsto \Ro, \\
a_\mathcal{D} &: H^2(\Tau_h) \times H^2(\Tau_h) \mapsto \Ro, \\
a_{\mathcal{M}_s} &: L^\infty(\Tau_h) \times H^2(\Tau_h) \times H^2(\Tau_h) \mapsto \Ro, \\
a_{\mathcal{D},\Gamma^\mathrm{in}}  &:   H^2(\Tau_h) \times  H^2(\Tau_h)  \mapsto \Ro, \\
b_\mathcal{D} &: H^2(\Tau_h) \mapsto \Ro. 
\end{align*}
The forms used for the advection terms are expressed as
% TWO COLM FORMAT
%\begin{align*}
%a_{\mathcal{A}}(\vel;w,\vartheta)  =&   - \sum_{E \in \Tau_h} \int_E w \vel \cdot \nabla \vartheta  \notag \\
%&+ \sum_{e \in \Gamma_h \cup \Gamma^\mathrm{out}} \int_e w^\uparrow  \eavg{\vel \cdot \n_e} \edif{\vartheta},\\
%b_{\mathcal{A}}(\vel,w;\vartheta) =&  - \sum_{e \in \Gamma^\mathrm{in}} \int_e w \vel \cdot \n_e \vartheta,
%\end{align*}
%where the upwind term $w^\uparrow$ for the scalar quantity $w$ on the face $e$ is given by
%\begin{align*}
%w^\uparrow \big|_{e \in \Gamma_h} =& \begin{cases} w|_{E_{k^-}} \quad & \text{if } \eavg{\vel} \cdot \n_e \geq 0,\\
%								          w|_{E_{k^+}} \quad & \text{if } \eavg{\vel} \cdot \n_e < 0,	
%			 			     \end{cases} \\
%w^\uparrow \big|_{e \in \dbnd} =& \begin{cases} w|_{E_{k^-}} \quad & \text{if } \vel \cdot \n_e \geq 0,\\
%								  0 \quad & \text{if } \vel \cdot \n_e < 0.
%			 			     \end{cases}				          
%\end{align*}
% SINGLE COLM FORMAT
\begin{align*}
a_{\mathcal{A}}(\vel;w,\vartheta)  =&   - \sum_{E \in \Tau_h} \int_E w \vel \cdot \nabla \vartheta  + \sum_{e \in \Gamma_h \cup \Gamma^\mathrm{out}} \int_e w^\uparrow  \eavg{\vel \cdot \n_e} \edif{\vartheta},\\
b_{\mathcal{A}}(\vel,w;\vartheta) =&  - \sum_{e \in \Gamma^\mathrm{in}} \int_e w \vel \cdot \n_e \vartheta,
\end{align*}
where the upwind term $w^\uparrow$ for the scalar quantity $w$ on the face $e$ is given by
% TWO COLM FORMAT
%\begin{align*}
%w^\uparrow \big|_{e \in \Gamma_h} =& \begin{cases} w|_{E_{k^-}} \quad & \text{if } \eavg{\vel} \cdot \n_e \geq 0,\\
%								          w|_{E_{k^+}} \quad & \text{if } \eavg{\vel} \cdot \n_e < 0,	
%			 			     \end{cases} \\
%w^\uparrow \big|_{e \in \dbnd} =& \begin{cases} w|_{E_{k^-}} \quad & \text{if } \vel \cdot \n_e \geq 0,\\
%								  0 \quad & \text{if } \vel \cdot \n_e < 0.
%			 			     \end{cases}				          
%\end{align*}
% SINGLE COLM FORMAT
\begin{align*}
w^\uparrow \big|_{e \in \Gamma_h} = \begin{cases} w|_{E_{k^-}} \quad & \text{if } \eavg{\vel} \cdot \n_e \geq 0,\\
								          w|_{E_{k^+}} \quad & \text{if } \eavg{\vel} \cdot \n_e < 0,	
			 			     \end{cases} \qquad
						     w^\uparrow \big|_{e \in \dbnd} = \begin{cases}
						      w|_{E_{k^-}} \quad & \text{if } \vel \cdot \n_e \geq 0,\\
								  0 \quad & \text{if } \vel \cdot \n_e < 0.
			 			     \end{cases}				          
\end{align*}

The forms corresponding to the diffusion terms are given by
% TWO COLM FORMAT
%\begin{align*}
%&a_\mathcal{D}(w,\vartheta)=  \sum_{E \in \Tau_h} \int_E  \nabla w \cdot \nabla \vartheta  - \sum_{e \in \Gamma_h} \int_e \left( \eavg{ \nabla w \cdot \n_e} \edif{\vartheta}\right)  \notag\\
%& \qquad  - \sum_{e \in \Gamma_h} \int_e \left( \eavg{ \nabla \vartheta \cdot \n_e} \edif{w}\right) + \frac{\sigma_\mathcal{D}}{h} \sum_{e \in \Gamma_h} \int_e \edif{w} \edif{\vartheta},\\
%%%%
%&a_{\mathcal{M}_s}(z;w,\vartheta)=  \sum_{E \in \Tau_h} \int_E M_s(z) \nabla w \cdot \nabla \vartheta  \notag \\
%& \qquad - \sum_{e \in \Gamma_h} \int_e \left( \eavg{M_s(z) \nabla w \cdot \n_e} \edif{\vartheta}\right)  \notag\\
%& \qquad + \sum_{e \in \Gamma_h} \int_e \left( \eavg{M_s(z) \nabla \vartheta \cdot \n_e} \edif{w}\right)  \notag \\
%& \qquad + \frac{\sigma_\mathcal{M}}{h} \sum_{e \in \Gamma_h} \int_e \edif{w} \edif{\vartheta}, \\
%&a_{\mathcal{D},\Gamma^\mathrm{in}}(w,\vartheta)=   - \sum_{e \in \Gamma^\mathrm{in}} \int_e \left( \nabla w \cdot \n_e \right) \vartheta \notag \\
%&\qquad - \sum_{e \in \Gamma^\mathrm{in}} \int_e \left( \nabla \vartheta \cdot \n_e \right) w + \frac{\sigma_{\mathcal{D}_1}}{h} \sum_{e \in \Gamma^\mathrm{in}} \int_e w \vartheta,\\
%&\ddif{\vartheta}= - \sum_{e \in \Gamma^\mathrm{in}} \int_e \left( \nabla \vartheta \cdot \n_e \right) c_\mathrm{in}  + \frac{\sigma_{\mathcal{D}_1}}{h} \sum_{e \in \Gamma^\mathrm{in}} \int_e c_\mathrm{in} \vartheta.
%\end{align*}
% SINGLE COLM FORMAT
\begin{align*}
a_\mathcal{D}(w,\vartheta)=&  \sum_{E \in \Tau_h} \int_E  \nabla w \cdot \nabla \vartheta  - \sum_{e \in \Gamma_h} \int_e \left( \eavg{ \nabla w \cdot \n_e} \edif{\vartheta}\right) 
 - \sum_{e \in \Gamma_h} \int_e \left( \eavg{ \nabla \vartheta \cdot \n_e} \edif{w}\right) 
+ \frac{\sigma_\mathcal{D}}{h} \sum_{e \in \Gamma_h} \int_e \edif{w} \edif{\vartheta}, \\
%%%
a_{\mathcal{M}_s}(z;w,\vartheta)=&  \sum_{E \in \Tau_h} \int_E M_s(z) \nabla w \cdot \nabla \vartheta  
- \sum_{e \in \Gamma_h} \int_e \left( \eavg{M_s(z) \nabla w \cdot \n_e} \edif{\vartheta}\right) 
 + \sum_{e \in \Gamma_h} \int_e \left( \eavg{M_s(z) \nabla \vartheta \cdot \n_e} \edif{w}\right) \\
&+ \frac{\sigma_\mathcal{M}}{h} \sum_{e \in \Gamma_h} \int_e \edif{w} \edif{\vartheta},\\
a_{\mathcal{D},\Gamma^\mathrm{in}}(w,\vartheta)=&   - \sum_{e \in \Gamma_\mathrm{in}} \int_e \left( \nabla w \cdot \n_e \right) \vartheta 
- \sum_{e \in \Gamma_\mathrm{in}} \int_e \left( \nabla \vartheta \cdot \n_e \right) w 
+ \frac{\sigma_{\mathcal{D}_1}}{h} \sum_{e \in \Gamma_\mathrm{in}} \int_e w \vartheta,\\
\ddif{\vartheta}=& - \sum_{e \in \Gamma_\mathrm{in}} \int_e \left( \nabla \vartheta \cdot \n_e \right) c_\mathrm{in} 
+ \frac{\sigma_{\mathcal{D}_1}}{h} \sum_{e \in \Gamma_\mathrm{in}} \int_e c_\mathrm{in} \vartheta,
\end{align*}
We point out that the form $a_{\mathcal{D}}$ is a symmetric bilinear form whereas the form $a_{\mathcal{M}_s}$ is non-symmetric.
This choice has been carefully made to produce a scheme that would be energy dissipative according to Proposition~\ref{thm:energy_decay}.

\subsection{Fully-discrete scheme}
The spatial discretization of \eqref{eqn:model} is performed using IPDG. We follow closely the formulation considered for the advective pure Cahn-Hilliard system considered in \cite{FRANK18}. Define $\pspace_q(E)$ to be set of all polynomials on $E$ of degree at most $q$ and
define the broken polynomial space
\[
\pspace_q(\Tau_h) = \prod_{E_k \in \Tau_h} \pspace_q(E_k).
\]

Let $\vel_h^n$ denote the $L^2$ projection of $\vel^n$ into $\pspace_q(\Tau_h)^3$. Using $(\cdot,\cdot)$ to denote the $L^2$ inner-product on $\Omega$, we consider the following fully-discrete scheme for the temporal algorithm described in Section \ref{sec:temporal_disc}: 

Given $(c_h^{n-1}, s_h^{n-1})$ find $c^n_h,s^n_h,\cpch^n,\cpsh^n \in \pspace_q(\Tau_h)$ such that for all $\vartheta_h\in \pspace_q(\Tau_h)$
%TWO COLM FORMAT
%\begin{subequations}
%\label{eqn:fd_scheme}
%\begin{align}
%&(c_h^n,\vartheta_h) + \frac{\tau_n \, M_c}{\Pe_c} a_{\mathcal{D}}(\mu_{ch}^n,\vartheta_h) 
%+ \tau_n \, a_{\mathcal{A}}(\vel^{n}_h;c_h^n,\vartheta_h) \notag \\
%& \qquad = (c^{n-1}_h,\vartheta_h)
%+ \tau_n \, b_{\mathcal{A}}(\vel_h^n,c_\mathrm{in};\vartheta_h), \label{eqn:fd_schemea}\\
%&(s_h^n,\vartheta_h) + \frac{\tau_n}{\Pe_s} a_{\mathcal{M}_s}(s_h^{n-1};\mu_{sh}^n, \vartheta_h) +\tau_n \, a_{\mathcal{A}}(\vel^{n}_h;s_h^n,\vartheta_h)  \notag \\
%&\qquad =(s_h^{n-1},\vartheta_h) + \tau_n b_{\mathcal{A}}(\vel_h^n,s_\mathrm{in};\vartheta_h), \label{eqn:fd_schemeb}\\
%& -(\mu_{ch}^n,\vartheta_h) +   (\Phi_+^\prime(c^n_h),\vartheta_h)  + \alpha_{4} (c^{n-1}_h +c_h^n, s^n_h\,\vartheta_h)\notag \\
%&\qquad + \Cn^2 \left(a_\mathcal{D}(c^n_h,\vartheta_h) 
% + a_{\mathcal{D},\Gamma^\mathrm{in}}(c^n_h,\vartheta_h)\right)  \notag \\
%&\qquad  -\alpha_{3} (\Phi_+^\prime(c^{n-1}_h)+\Phi_-^\prime(c^{n}_h),s_h^n\,\vartheta_h)  \notag\\
%& \qquad  = -  (\Phi_-^\prime(c^{n-1}_h),\vartheta_h) + \Cn^2\ddif{\vartheta_h},\label{eqn:fd_schemec}\\
%&-(\mu_{sh}^n,\vartheta_h)  + \alpha_2 (\Psi^\prime(s^n_h),\vartheta_h) \notag \\
%&\qquad =  \alpha_{3} (\Phi(c^{n-1}_h),\vartheta_h)
%-\alpha_{4} ((c^{n-1}_h)^2,\vartheta_h). \label{eqn:fd_schemed} 
%\end{align}
%\end{subequations}
%SINGLE COLM FORMAT
\begin{subequations}
\label{eqn:fd_scheme}
\begin{align}
(c_h^n,\vartheta_h) + \frac{\tau_n \, M_c}{\Pe_c} a_{\mathcal{D}}(\mu_{ch}^n,\vartheta_h) 
+ \tau_n \, a_{\mathcal{A}}(\vel^{n}_h;c_h^n,\vartheta_h) &= (c^{n-1}_h,\vartheta_h)
+ \tau_n \, b_{\mathcal{A}}(\vel_h^n,c_\mathrm{in};\vartheta_h), \label{eqn:fd_schemea}\\
(s_h^n,\vartheta_h) + \frac{\tau_n}{\Pe_s} a_{\mathcal{M}_s}(s_h^{n-1};\mu_{sh}^n, \vartheta_h) 
+\tau_n \, a_{\mathcal{A}}(\vel^{n}_h;s_h^n,\vartheta_h)  
&=(s_h^{n-1},\vartheta_h)
+ \tau_n b_{\mathcal{A}}(\vel_h^n,s_\mathrm{in};\vartheta_h), \label{eqn:fd_schemeb}\\
-(\mu_{ch}^n,\vartheta_h) +   (\Phi_+^\prime(c^n_h),\vartheta_h)  + \Cn^2 \left(a_\mathcal{D}(c^n_h,\vartheta_h) 
 + a_{\mathcal{D},\Gamma^\mathrm{in}}(c^n_h,\vartheta_h)\right) &\\
 -\alpha_{3} (\Phi_+^\prime(c^{n-1}_h)+\Phi_-^\prime(c^{n}_h),s_h^n\,\vartheta_h)  
+ \alpha_{4} (c^{n-1}_h +c_h^n, s^n_h\,\vartheta_h)
&= -  (\Phi_-^\prime(c^{n-1}_h),\vartheta_h) 
 \quad + \Cn^2\ddif{\vartheta_h},\label{eqn:fd_schemec}\\
-(\mu_{sh}^n,\vartheta_h)  + \alpha_2 (\Psi^\prime(s^n_h),\vartheta_h) &=  \alpha_{3} (\Phi(c^{n-1}_h),\vartheta_h)
-\alpha_{4} ((c^{n-1}_h)^2,\vartheta_h). \label{eqn:fd_schemed} 
\end{align}
\end{subequations}
We finish this section by recalling a trace inequality and a property satisfied by the form $a_\mathcal{D}$ that is used in the next
proposition on the total discrete energy.
%We prove the the operator $\adif{.}{.}{.}$ satisfies a coercivity type result. The proof requires the following modified version of the trace-inequality.
\begin{lemma}\label{lem:trace_ineq}
Let $E$ be a triangle or rectangle in 2D, a tetrahedron or a parallelepiped in 3D. Let $v \in \pspace_q(E)$, $z \in \pspace_{\hat {q}}(E)$. Then there exists a constant $C_t$ depending only on $q$ and $\hat {q}$ such that for all $e \in \partial E$
\begin{equation}\label{eqn:trace_ineq}
\|z \nabla v \cdot \n_e \|_{L^2(e)} \leq C_t |e|^{1/2} |E|^{-1/2} \| z \nabla v \|_{L^2(E)},
\end{equation}
where $|E|$ (resp. $|e|$) denotes the measure of $E$ (resp. $e$).
\end{lemma}
We now recall positivity results for $a_{\mathcal{M}_s}$ and $a_{\mathcal{D}}$ \cite{riviere2008}.
\begin{lemma}\label{lem:coerc}
Let $\vartheta_h$ and $z_h$ be in $\pspace_q(\Tau_h)$ 
for integer $q\geq 1$. We have
\[
a_{\mathcal{M}_s}(z_h;\vartheta_h,\vartheta_h) \geq 0.
\]
Let $N_0$ denotes the maximum number of neighbours an element can have and  
assume that the penalty parameter $\sigma_\mathcal{D}$ is large enough, namely
\begin{equation}\label{eq:sigmabound} 
\sigma_{\mathcal{D}} \geq  4 C_t^2 N_0. 
\end{equation}
Then we have
\[
a_\mathcal{D}(\vartheta_h,\vartheta_h) \geq 0.
\]
\end{lemma}

\subsection{Discrete energy stability}\label{sec:disc_energy_decay}
We now show that for the closed non-advective system, the discrete free energy decays in a consistent manner under the assumption
that the numerical approximation of the surfactant remains non-negative. While the maximum principle cannot
be obtained theoretically for the discontinuous Galerkin solution, the following proposition states
an important property for physical systems.
Numerical results in Section~\ref{sec:results} show the decay of the numerical energy and confirm the theoretical result. 

The discrete total energy at time $t_n$ is defined by
%TWO COLM
%\begin{equation}\label{eqn:discrete_energy}
%\begin{aligned}
%\mathcal{F}^n_h =& (\Phi(c^n_h),1) 
%+ \frac{\mbox{Cn}^2}{2} a_\mathcal{D}(c^n_h,c_h^n) 
%+ \alpha_2 (\Psi(s^n_h),1) \notag \\ 
%&- \alpha_{3} (\Phi(c^n_h),s^n_h) 
%+  \alpha_{4} ((c^n_h)^2, s_h^n).
%\end{aligned}
%\end{equation}
%ONE COLM
\begin{equation}\label{eqn:discrete_energy}
\mathcal{F}^n_h = (\Phi(c^n_h),1) 
+ \frac{\mbox{Cn}^2}{2} a_\mathcal{D}(c^n_h,c_h^n) 
+ \alpha_2 (\Psi(s^n_h),1) - \alpha_{3} (\Phi(c^n_h),s^n_h) 
+  \alpha_{4} ((c^n_h)^2, s_h^n).
\end{equation}
\begin{proposition}\label{thm:energy_decay}
Assume that $\vel = {\bf 0}$ and assume that $\sigma_\mathcal{D}$ satisfies \eqref{eq:sigmabound}.
Assume that the numerical approximation for the surfactant is non-negative ($s_h^n \geq 0$).
Then the scheme \eqref{eqn:fd_scheme} ensures the decay of total free-energy: 
\begin{equation}\label{eqn:disc_energy_decay}
\mathcal{F}^n_h \leq \mathcal{F}^{n-1}_h, \quad \forall 1\leq n\leq N_T.
\end{equation}
\end{proposition}
\begin{proof}
%For readibility, we define
%%\[
%\delta_\tau c_h^n = \frac{c_h^n-c_h^{n-1}}{\tau_n}, \quad
%\delta_\tau s_h^n = \frac{s_h^n-s_h^{n-1}}{\tau_n}.
%\]
We choose $\vartheta_h = \mu_{ch}^n$ in \eqref{eqn:fd_schemea}, 
$\vartheta_h = \mu_{sh}^n$ in \eqref{eqn:fd_schemeb}, 
$\vartheta_h =  c_h^n-c_h^{n-1}$ in \eqref{eqn:fd_schemec} 
and $\vartheta_h = s_h^n-s_h^{n-1}$ in \eqref{eqn:fd_schemed}. We add the resulting equations and obtain
% TWO COLM
%\begin{align}
%&\left(\Phi_+'(c_h^n)+\Phi_-'(c_h^{n-1}), c_h^n-c_h^{n-1}\right) \notag\\
%&\ +\Cn^2 a_\mathcal{D}(c_h^n,c_h^n-c_h^{n-1})+\alpha_2\left(\Psi'(s_h^n),s_h^n-s_h^{n-1}\right)\nonumber\\
%&\ -\alpha_{3} \left(\Phi_+'(c_h^{n-1})+\Phi_-'(c_h^n), s_h^n(c_h^n-c_h^{n-1})\right) \notag\\
%&\ -\alpha_{3} (\Phi(c_h^{n-1}), s_h^n-s_h^{n-1})
%+\alpha_{4} (s_h^n, (c_h^n)^2-(c_h^{n-1})^2) \notag\\
%&\ +\alpha_{4} (s_h^n-s_h^{n-1}, (c_h^{n-1})^2)\nonumber\\
%&\ =-\frac{\tau_n \, M_c}{\Pe_c} a_\mathcal{D}(\mu_{ch}^n,\mu_{ch}^n)
%-\frac{\tau_n}{\Pe_s} a_{\mathcal{M}_s}(s_h^{n-1};\mu_{sh}^n,\mu_{sh}^n) \notag \\
%& \ \leq 0,\label{eq:in1}
%\end{align}
% ONE COLM
\begin{align}
\left(\Phi_+'(c_h^n)+\Phi_-'(c_h^{n-1}), c_h^n-c_h^{n-1}\right)
-\alpha_{3} \left(\Phi_+'(c_h^{n-1})+\Phi_-'(c_h^n), s_h^n(c_h^n-c_h^{n-1})\right)
+\alpha_2\left(\Psi'(s_h^n),s_h^n-s_h^{n-1}\right)\nonumber\\
+\Cn^2 a_\mathcal{D}(c_h^n,c_h^n-c_h^{n-1})
-\alpha_{3} (\Phi(c_h^{n-1}), s_h^n-s_h^{n-1})
+\alpha_{4} (s_h^n, (c_h^n)^2-(c_h^{n-1})^2)
+\alpha_{4} (s_h^n-s_h^{n-1}, (c_h^{n-1})^2)\nonumber\\
=-\frac{\tau_n \, M_c}{\Pe_c} a_\mathcal{D}(\mu_{ch}^n,\mu_{ch}^n)
-\frac{\tau_n}{\Pe_s} a_{\mathcal{M}_s}(s_h^{n-1};\mu_{sh}^n,\mu_{sh}^n)
\leq 0,\label{eq:in1}
\end{align}
thanks to Lemma~\ref{lem:coerc}.

Using Taylor expansions, there exist $\xi_1, \xi_2, \xi_3, \xi_4$ between $c^{n-1}_h$ and $c^{n}_h$ 
and $\xi_5$ between $s_h^{n-1}$ and $s_h^n$ such that
%TWO COLM
%\begin{subequations}\label{eqn:fd3}
%\begin{align}
%\Phi_+^\prime(c^n_h)(c^n_h - c^{n-1}_h)  =& \ \Phi_+(c^n_h) - \Phi_+(c^{n-1}_h) \notag \\ 
%&+\frac{1}{2} \Phi_+^{\prime \prime}(\xi_1)(c^n_h - c^{n-1}_h)^2, \label{eqn:fd3a}\\ 
%\Phi_-^\prime(c^{n-1}_h)(c^n_h - c^{n-1}_h)  =& \ \Phi_-(c^n_h) - \Phi_-(c^{n-1}_h) \notag \\ 
%& -\frac{1}{2} \Phi_-^{\prime \prime}(\xi_2)(c^n_h - c^{n-1}_h)^2,\label{eqn:fd3b}\\
%\Phi_+^\prime(c^{n-1}_h)(c^n_h - c^{n-1}_h)  =& \ \Phi_+(c^n_h) - \Phi_+(c^{n-1}_h) \notag \\ 
%& - \frac{1}{2} \Phi_+^{\prime \prime}(\xi_3)(c^n_h - c^{n-1}_h)^2,\label{eqn:fd3c}\\
%\Phi_-^\prime(c^n_h)(c^n_h - c^{n-1}_h)  =& \ \Phi_-(c^n_h) - \Phi_-(c^{n-1}_h) \notag \\ 
%& +\frac{1}{2} \Phi_-^{\prime \prime}(\xi_4)(c^n_h - c^{n-1}_h)^2,\label{eqn:fd3d}\\
%\Psi^\prime(s^n_h)(s^n_h - s^{n-1}_h)  =& \ \Psi(s^n_h) - \Psi(s^{n-1}_h) \notag \\ 
%& +\frac{1}{2} \Psi^{\prime \prime}(\xi_5)(s^n_h - s^{n-1}_h)^2. \label{eqn:fd3e}
%\end{align}
%\end{subequations}
%ONE COLM
\begin{subequations}\label{eqn:fd3}
\begin{align}
\Phi_+^\prime(c^n_h)(c^n_h - c^{n-1}_h) & = \Phi_+(c^n_h) - \Phi_+(c^{n-1}_h) +\frac{1}{2} \Phi_+^{\prime \prime}(\xi_1)(c^n_h - c^{n-1}_h)^2, \label{eqn:fd3a}\\ 
\Phi_-^\prime(c^{n-1}_h)(c^n_h - c^{n-1}_h) & = \Phi_-(c^n_h) - \Phi_-(c^{n-1}_h) -\frac{1}{2} \Phi_-^{\prime \prime}(\xi_2)(c^n_h - c^{n-1}_h)^2,\label{eqn:fd3b}\\
\Phi_+^\prime(c^{n-1}_h)(c^n_h - c^{n-1}_h) & = \Phi_+(c^n_h) - \Phi_+(c^{n-1}_h) - \frac{1}{2} \Phi_+^{\prime \prime}(\xi_3)(c^n_h - c^{n-1}_h)^2,\label{eqn:fd3c}\\
\Phi_-^\prime(c^n_h)(c^n_h - c^{n-1}_h) & = \Phi_-(c^n_h) - \Phi_-(c^{n-1}_h) +\frac{1}{2} \Phi_-^{\prime \prime}(\xi_4)(c^n_h - c^{n-1}_h)^2,\label{eqn:fd3d}\\
\Psi^\prime(s^n_h)(s^n_h - s^{n-1}_h) & = \Psi(s^n_h) - \Psi(s^{n-1}_h) +\frac{1}{2} \Psi^{\prime \prime}(\xi_5)(s^n_h - s^{n-1}_h)^2. \label{eqn:fd3e}
\end{align}
\end{subequations}
Since $\Phi_+$ is convex and $\Phi_-$ is concave, we have with \eqref{eqn:fd3a} and \eqref{eqn:fd3b}
% TWO COLM
%\begin{align*}
%&\left(\Phi_+'(c_h^n)+\Phi_-'(c_h^{n-1}), c_h^n-c_h^{n-1}\right)\\
%&\qquad= \left(\Phi(c_h^n)-\Phi(c_h^{n-1}),1\right)
%+ \frac12 \left(\Phi_+''(\xi_1), (c_h^n-c_h^{n-1})^2\right) \\
%&\ \  \qquad- \frac12 \left(\Phi_-''(\xi_2), (c_h^n-c_h^{n-1})^2\right)\\
%&\qquad\geq \left(\Phi(c_h^n)-\Phi(c_h^{n-1}),1\right).
%\end{align*}
%ONE COLM
\begin{eqnarray*}
\left(\Phi_+'(c_h^n)+\Phi_-'(c_h^{n-1}), c_h^n-c_h^{n-1}\right)
&=& \left(\Phi(c_h^n)-\Phi(c_h^{n-1}),1\right)
+ \frac12 \left(\Phi_+''(\xi_1), (c_h^n-c_h^{n-1})^2\right)
- \frac12 \left(\Phi_-''(\xi_2), (c_h^n-c_h^{n-1})^2\right)
\\
&\geq& \left(\Phi(c_h^n)-\Phi(c_h^{n-1}),1\right).
\end{eqnarray*}
Similarly, with \eqref{eqn:fd3b}, \eqref{eqn:fd3c} and the assumption $s_h\geq 0$, we have
% TWO COLM
%\begin{align*}
%&-\alpha_{3} \left(\Phi_+'(c_h^{n-1})+\Phi_-'(c_h^n), s_h^n(c_h^n-c_h^{n-1})\right) \\
%&\qquad \geq -\alpha_3 \left(\Phi(c_h^n)-\Phi(c_h^{n-1}),s_h^n\right),
%\end{align*}
% ONE COLM
\begin{align*}
-\alpha_{3} \left(\Phi_+'(c_h^{n-1})+\Phi_-'(c_h^n), s_h^n(c_h^n-c_h^{n-1})\right)  \geq -\alpha_3 \left(\Phi(c_h^n)-\Phi(c_h^{n-1}),s_h^n\right),
\end{align*}
and since $\Psi$ is convex,  with \eqref{eqn:fd3e}, we have
\[
\alpha_2\left(\Psi'(s_h^n),s_h^n-s_h^{n-1}\right) \geq \alpha_2 \left(\Psi(s_h^n)-\Psi(s_h^{n-1}),1\right).
\]
The inequality \eqref{eq:in1} simplifies to:
% TWO COLM
%\begin{align*}
%&\left(\Phi(c_h^n)-\Phi(c_h^{n-1}), 1\right)
%-\alpha_{3} \left((\Phi(c_h^{n}), s_h^n)-(\Phi(c_h^{n-1}),s_h^{n-1})\right) \\
%&\qquad+\alpha_2\left(\Psi(s_h^n)-\Psi(s_h^{n-1}),1\right) +\Cn^2 a_\mathcal{D}(1;c_h^n,c_h^n-c_h^{n-1}) \\
%& \qquad+\alpha_{4} \left( (s_h^n, (c_h^n)^2)-(s_h^{n-1},(c_h^{n-1})^2)\right)
%\leq 0.
%\end{align*}
%ONE COLM
\begin{align*}
\left(\Phi(c_h^n)-\Phi(c_h^{n-1}), 1\right)
-\alpha_{3} \left((\Phi(c_h^{n}), s_h^n)-(\Phi(c_h^{n-1}),s_h^{n-1})\right)
+\alpha_2\left(\Psi(s_h^n)-\Psi(s_h^{n-1}),1\right)\\
+\Cn^2 a_\mathcal{D}(1;c_h^n,c_h^n-c_h^{n-1})
+\alpha_{4} \left( (s_h^n, (c_h^n)^2)-(s_h^{n-1},(c_h^{n-1})^2)\right)
\leq 0.
\end{align*}
Since the form $a_{\mathcal{D}}(\cdot,\cdot)$ is symmetric and bilinear, we have
\[
\frac{1}{2}a_{\mathcal{D}}(c_h^n,c_h^n)
-\frac12 a_{\mathcal{D}}(c_h^{n-1},c_h^{n-1})
\leq a_\mathcal{D}(c_h^n,c_h^n-c_h^{n-1}).
\]
This bound with the one above concludes the proof.

\end{proof}

%\end{document}

%  % !  TEX root = main.tex
%\documentclass[main.tex]{subfiles}
%
%\begin{document}

\section{Numerical Results}\label{sec:results}
\label{sec:results}
We demonstrate the performance of the proposed IPDG scheme, by using it to solve a number of problems with varying complexity. Piecewise linear approximation spaces are used with the penalty parameters set as $\sigma_\mathcal{D} = 2.0$, $\sigma_\mathcal{M} =2.0$ and $\sigma_{\mathcal{D}_1} = 8.0$. We use the tensor product of one dimensional Legendre polynomials to form the basis in each element. Following the strategy of \cite{FRANK18}, the implicit system describing the scheme \eqref{eqn:fd_scheme} is reduced using Schur complement to a smaller system solving for $c^n_h$ and $s^n_h$. The reduced system is solved using a Newton's method, followed by a direct computation of $\mu_{ch}^n$ and $\mu_{sh}^n$. In all experiments, we choose $\Cn = h$ and a uniform time-step $\tau_n = \tau = 10^{-3}$, unless specified otherwise. In all two and three-dimensional plots for the order parameter, the phase corresponding to $c=1$ will be depicted in red, the phase corresponding to $c=-1$ will be depicted in blue, and the diffuse-interface by a steep color-gradient.

\subsection{Adsorption isotherm}
We begin by testing the capability of the numerical scheme to capture key physical properties of the underlying model at equilibrium. In particular, we consider the equilibrium adsorption isotherm which relates the surfactant concentration at the surface to the bulk surfactant concentration. The choice of the free energy terms in \eqref{eqn:helmenergy} plays a crucial role in designing schemes that can faithfully recover the isotherm curves \cite{SMAN06,DIAM96,LIU10,ENGBLOM13}.

We consider a one-dimensional planar interface problem and use the subscript notations '$i$' and '$b$' to denote quantities defined at the interface and the bulk, respectively. We consider a dilute solution regime characterised by a small bulk surfactant concentration, i.e., $s_b \ll 1$. In order to carry out the analysis and obtain analytical expressions of equilibrium solution, we assume that the order parameter profile is independent of the surfactant loading at equilibrium \cite{DIAM96,LIU10}. Under these assumption, the order parameter $c$ at equilibrium is given by
\begin{equation}\label{eqn:planar_c}
c(x) = \text{tanh}\left( \frac{x-x_o}{\sqrt{2} \Cn}\right),
\end{equation}
 centered at $x_o=0.5$. Note that \eqref{eqn:planar_c} is the steady-state solution of \eqref{eqn:model} in the absence of a surfactant. 
 
At equilibrium, the chemical potential attains a constant value in the whole domain. Equating the chemical potential for the surfactant $\cps$ in the bulk to the value at any point $x$ in the domain, and using the fact that $c_b = \pm 1$, we can derive the expression for the surfactant
\begin{equation}\label{eqn:eq_s}
s(x) = \frac{s_b}{s_b + (1-s_b)s_q(x)} \approx \frac{s_b}{s_b + s_q(x)},
\end{equation}
where
\begin{equation}\label{eqn:sq}
s_q(x) = \exp \left[ -\frac{1}{\alpha_2} \big( \alpha_3 \Phi(c(x)) + \alpha_4(1-c(x)^2)\big)\right].
\end{equation}
A detailed derivation of these expression can be found in \cite{ENGBLOM13}. Evaluating \eqref{eqn:eq_s} at the interface and noting that $c_i = 0$, we get
\begin{equation}\label{eqn:isotherm}
\begin{aligned}
s_i &= \frac{s_b}{s_b + (1-s_b)s_{q,i}} \approx \frac{s_b}{s_b + s_{q,i}},\\
 s_{q,i} &= \exp \left[ -\frac{1}{\alpha_2} \big( \frac{\alpha_3}{4} + \alpha_4\big)\right].
\end{aligned}
\end{equation}
The relation \eqref{eqn:isotherm} is known as the \textit{Langmuir isotherm} with $s_{q,i}$ being the Langmuir adsorption constant.

We demonstrate that the DG scheme proposed in this work is able to recover the Langmuir isotherm. We consider the one-dimensional simulation on the domain $\Omega = (0,1)$ discretized using $N_{el} = 80$ elements. We set $\Pe_c = 1$, $\Pe_s = 1$, $\alpha_3 = 1.0$, $\alpha_4 = 0.25$ and $\Cn = 0.05 = 4h$. We consider three different isotherm curves by choosing $\alpha_2 \in \{0.1, 0.15, 0.2\}$ and $s_b \in [5\times 10^{-3},10^{-1}]$. The initial condition for the order parameter is set using \eqref{eqn:planar_c}, while the surfactant is prescribed by the shifted profile
\[
s(x) = \frac{s_b}{s_b + s_q(x-0.2)}.
\]
As shown in Figure~\ref{fig:steady_state_soln}, the surfactant profile diffuses to the interface at steady state and matches the equilibrium analytical expression \eqref{eqn:eq_s}. We also plot the numerically obtained values for $s_b$ versus $s_i$ in Figure \ref{fig:isotherms}, which clearly coincide with the analytical Langmuir isotherm curves given by \eqref{eqn:isotherm}.

\begin{figure}[htbp]
\begin{center}
\subfigure[$\alpha_2 = 0.1$]{\includegraphics[width=0.32\textwidth]{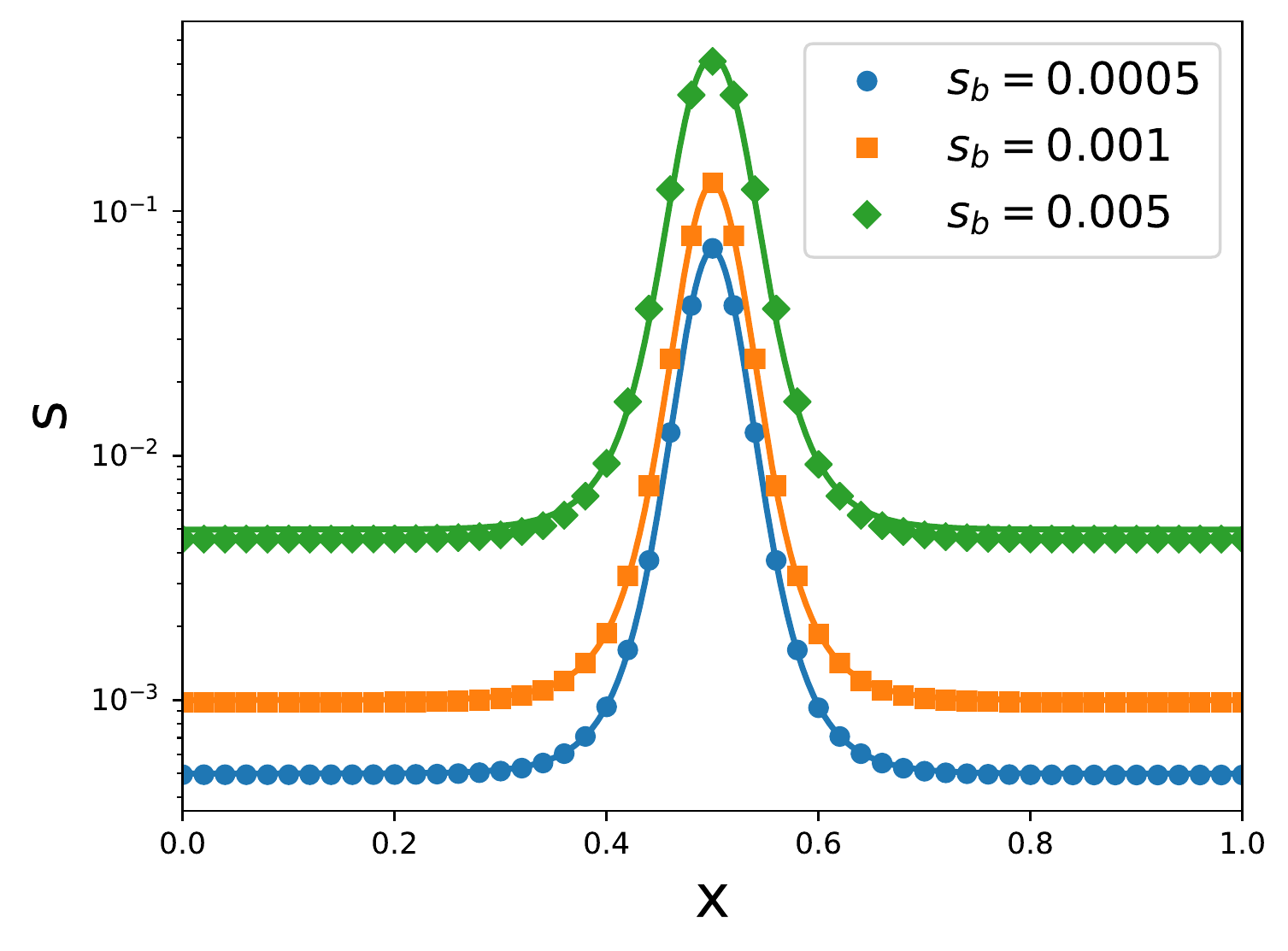}}
\subfigure[$\alpha_2 = 0.15$]{\includegraphics[width=0.32\textwidth]{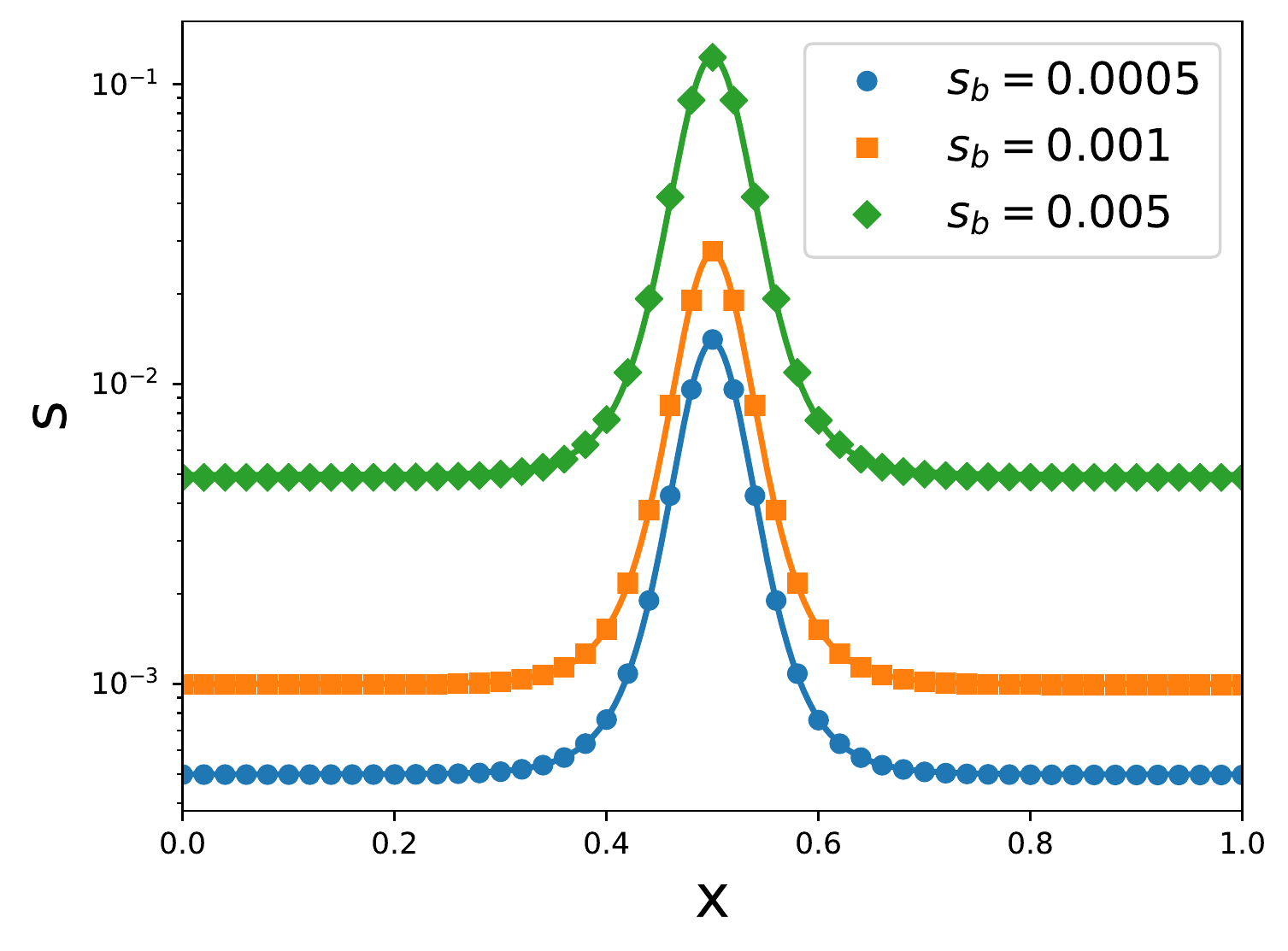}}
\subfigure[$\alpha_2 = 0.2$]{\includegraphics[width=0.32\textwidth]{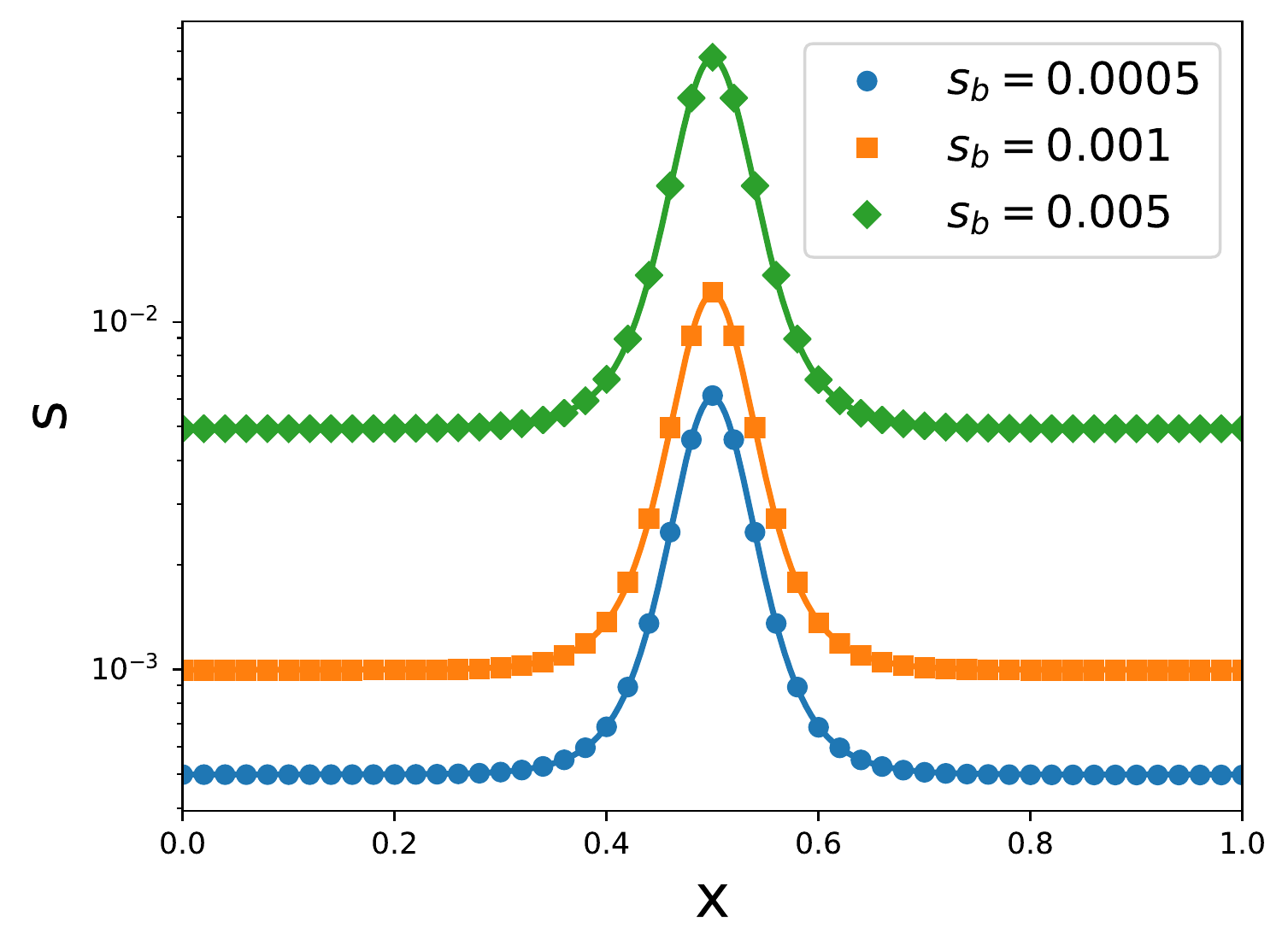}}
\caption{Equilibrium profiles for surfactant for varying values of $\alpha_2$ and $s_b$. The solid lines depict the analytical expression \eqref{eqn:eq_s}, while the markers show the numerical approximation}
\label{fig:steady_state_soln}
\end{center}
\end{figure}

\begin{figure}[htbp]
\begin{center}
\includegraphics[width=0.32\textwidth]{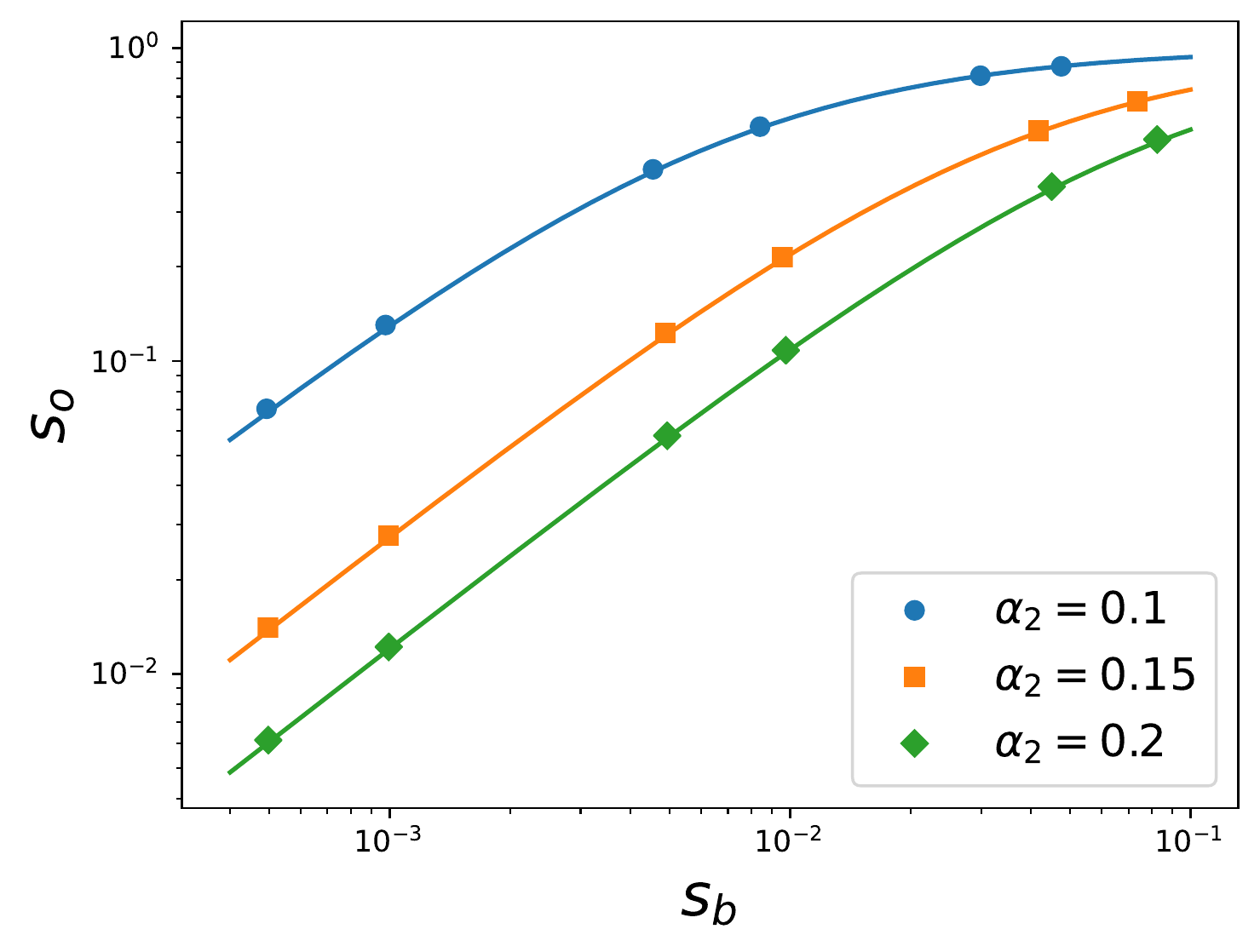}
\caption{Langmuir isotherms for $\alpha_2 \in \{0.1, 0.15, 0.2\}$. The solid lines depict the analytical expression \eqref{eqn:isotherm}, while the markers denote the values obtain from the numerical approximations}
\label{fig:isotherms}
\end{center}
\end{figure}

\subsection{Spinodal-drop interaction}
In order to  better highlight the diffusive dynamics of the order parameter in the presence of a surfactant, we consider a two dimensional non-advective problem where $c$ is initialized as a random constant on each element $E_k$:
\begin{equation}\label{eqn:spinodal_c_ic}
\begin{aligned}
c^0(x,y) \big|_{E_{k}} =  0.2 + 0.001 \omega_k, \quad \omega_k \in \text{rand}([-1,1]),
\end{aligned}
\end{equation}
while the surfactant is initialized as a circular drop
\begin{equation}\label{eqn:spinodal_s_ic}
\begin{aligned}
s^0(x,y)  &= \frac{1}{2} \left(0.5- 0.3\tanh\left(\frac{r_x - r_0}{\sqrt{2} \Cn}\right) \right), \\
 r_x &= \sqrt{(x-0.5)^2 + (y-0.5)^2},
\end{aligned}
\end{equation}
with $r_0 = 0.15$. The initial conditions are also shown in Figure \ref{fig:spinodal_init}. The boundary conditions are set by assuming the system to be closed, i.e., $\partial\Omega = \Gamma^{\mathrm{wall}}$. The domain $(0,1)^2$ is discretized using  $100\times 100$ square elements. The remaining parameters are set as $\Pe_c = 100$, $\Pe_s = 100$, $\alpha_2 = 1$,  $\alpha_3 = 1$ and $\alpha_4 = 1$.

\begin{figure}[htbp]
\begin{center}
\subfigure[Order-parameter]{\includegraphics[width=0.23\textwidth]{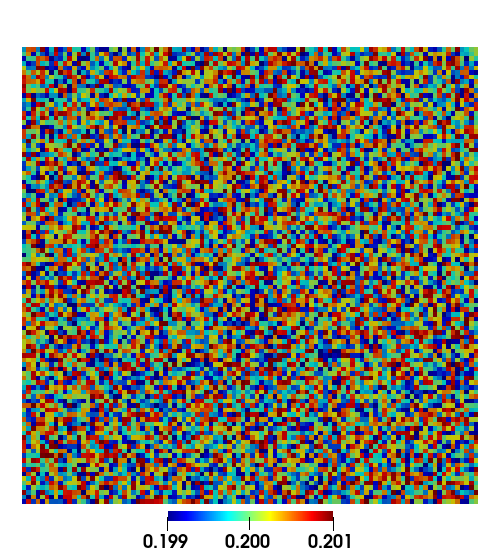}} 
\subfigure[Surfactant]{\includegraphics[width=0.23\textwidth]{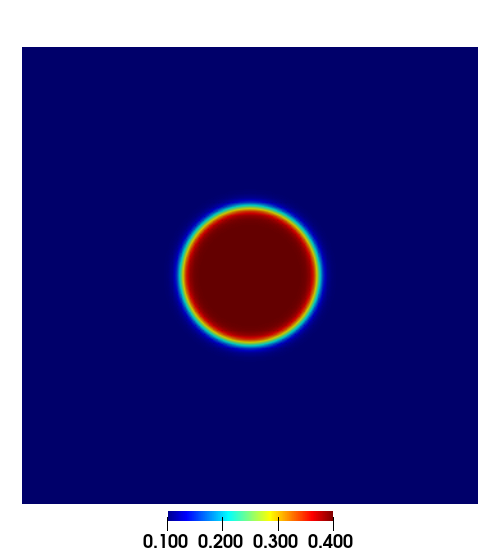}}
\caption{Initial conditions for the spinodal-drop problem: order parameter (left figure) and surfactant (right figure)}
\label{fig:spinodal_init}
\end{center}
\end{figure}

The evolving dynamics of the order parameter are depicted in Figure~\ref{fig:spinodal_c_zero_s} in the absence of any surfactant, i.e., $s\equiv 0$, while Figure~\ref{fig:spinodal_c_with_s} shows the evolution in the presence of a surfactant. The mixture moves towards a state of lower Helmholtz free energy, which is achieved via two key processes. Firstly, the contribution due to the interfacial energy is minimized by lowering the length of the diffusive interface. Thus, the smaller structures tend to coalesce together to form larger globules, i.e., coarsening, thereby reducing the total diffusive interface in the domain. Secondly, the free energy is reduced by forcing the surfactant to move to the diffusive interface and lowering its concentration in the bulk. This phenomena is depicted in Figure \ref{fig:spinodal_s}. We also note that the coalescence of the order parameter is more isotropic in the absence of a surfactant, while  the smaller drops coalesce along concentric circles when a drop surfactant is used. While it is expected that the order parameter will finally merge to a single bubble in both cases (if the simulation is run to steady state), the transient dynamics is strongly influenced by the surfactant.

In a closed system, the total amount of each of the three components is expected to be conserved.  
For the two components that form
two immiscible phases, a simple algebraic argument with the definition of the order parameter shows that conservation of each of the component is obtained by having the average quantity $\int_\Omega c$ constant.  For the surfactant, this also means that
$\int_\Omega s$ is constant.  This is also observed numerically, as shown in Figure \ref{fig:spinodal_mean}. We also show the decay of discrete free energy in Figure \ref{fig:spinodal_energy}. In the absence of surfactant, the free energy is only governed by $F_c$ (see \eqref{eqn:helmenergy}). While the total free energy decays in the presence of a surfactant, there is no guarantee that the individual contributors will decay in time, as can be seen in Figure \ref{fig:spinodal_energy}(b).

\begin{figure}[htbp]
\begin{center}
\subfigure[$t=0.5$]{\includegraphics[width=0.23\textwidth]{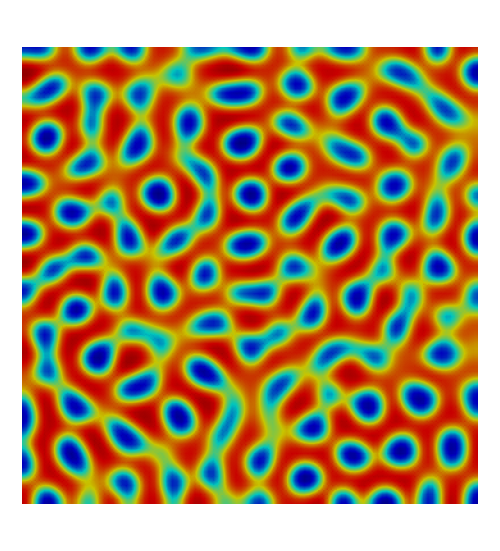}}
\subfigure[$t=1$]{\includegraphics[width=0.23\textwidth]{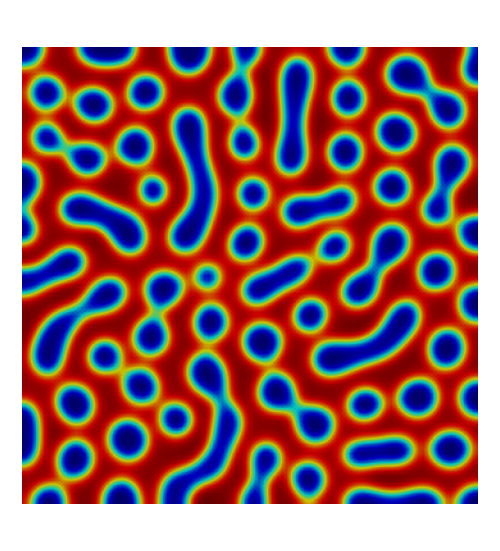}}
\subfigure[$t=5$]{\includegraphics[width=0.23\textwidth]{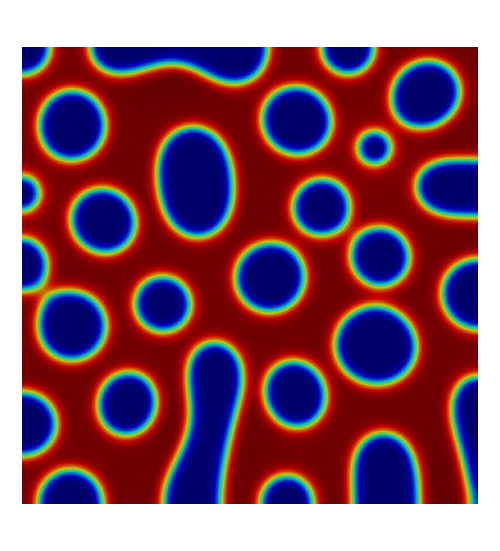}}
\subfigure[$t=50$]{\includegraphics[width=0.23\textwidth]{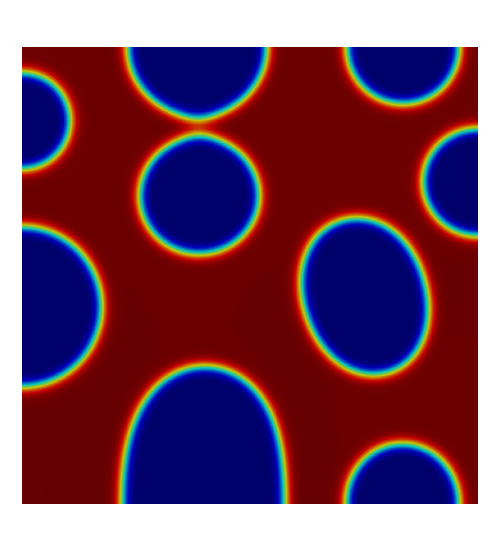}}\\
\caption{Evolution of the order parameter in the spinodal-drop problem in the absence of a surfactant}
\label{fig:spinodal_c_zero_s}
\end{center}
\end{figure}

\begin{figure}[htbp]
\begin{center}
\subfigure[$t=0.5$]{\includegraphics[width=0.23\textwidth]{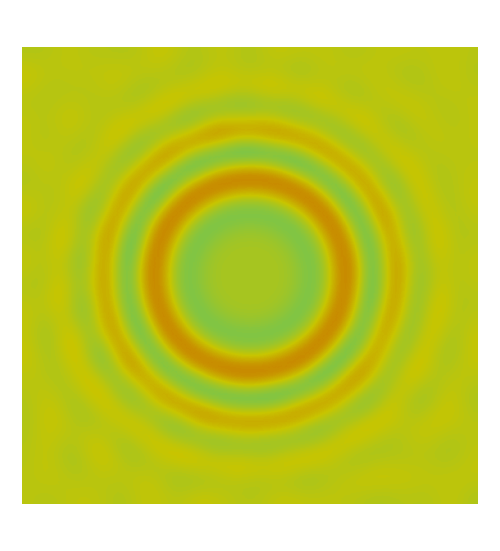}}
\subfigure[$t=1$]{\includegraphics[width=0.23\textwidth]{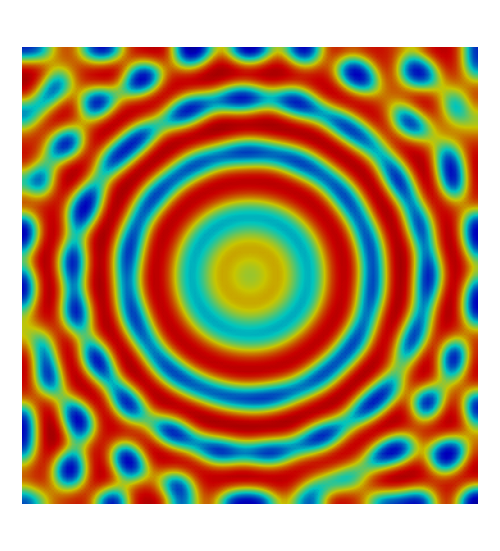}}
\subfigure[$t=5$]{\includegraphics[width=0.23\textwidth]{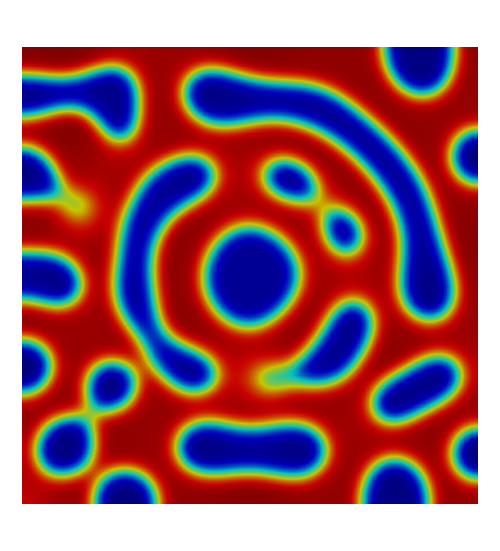}}
\subfigure[$t=50$]{\includegraphics[width=0.23\textwidth]{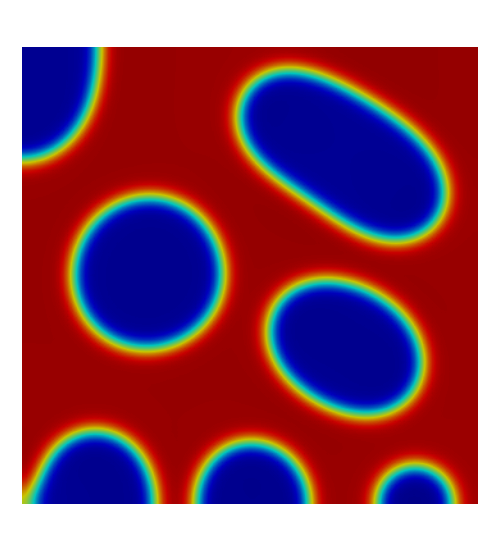}}\\
\caption{Evolution of the order parameter in the spinodal-drop problem in the presence of a surfactant}
\label{fig:spinodal_c_with_s}
\end{center}
\end{figure}

\begin{figure}[htbp]
\begin{center}
\subfigure[$t=0.5$]{\includegraphics[width=0.23\textwidth]{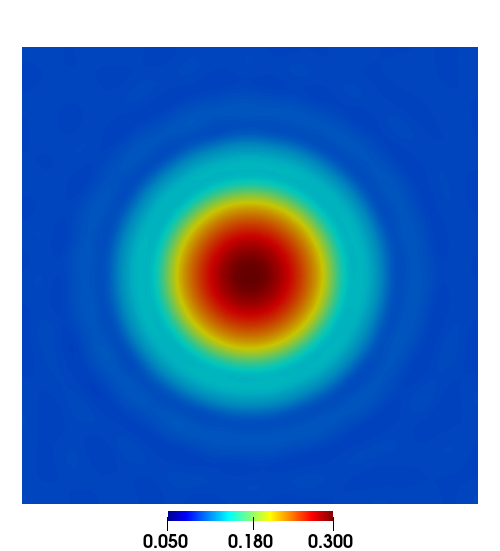}}
\subfigure[$t=1$]{\includegraphics[width=0.23\textwidth]{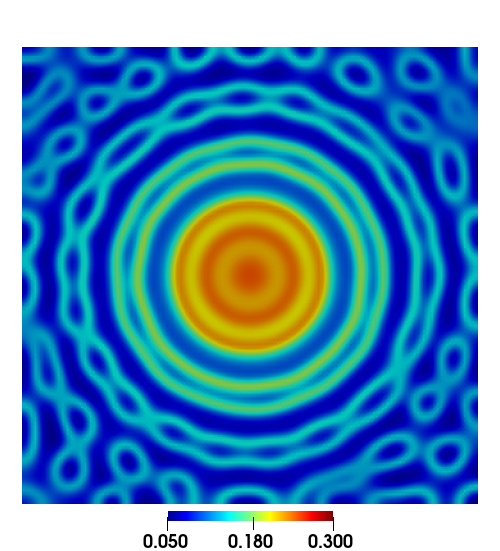}}
%\subfigure[$t=1.5$]{\includegraphics[width=0.23\textwidth]{Figures/Spinodal/Drop_s_1500}}\\
%\subfigure[$t=2.5$]{\includegraphics[width=0.23\textwidth]{Figures/Spinodal/Drop_s_2500}}
\subfigure[$t=5$]{\includegraphics[width=0.23\textwidth]{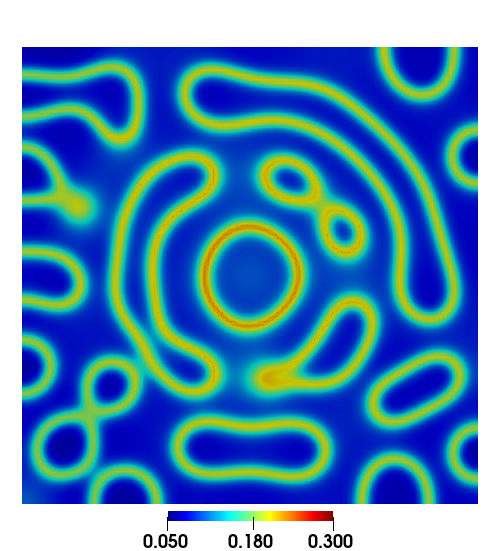}}
\subfigure[$t=50$]{\includegraphics[width=0.23\textwidth]{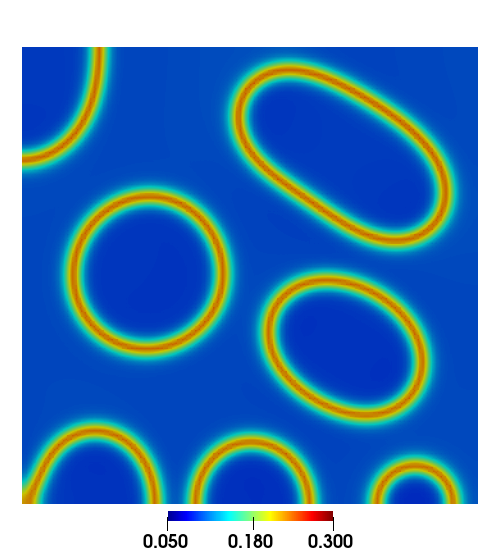}}\\
\caption{Evolution of the surfactant in the spinodal-drop problem}
\label{fig:spinodal_s}
\end{center}
\end{figure}

\begin{figure}[htbp]
\begin{center}
\subfigure[Mean $c$]{\includegraphics[width=0.32\textwidth]{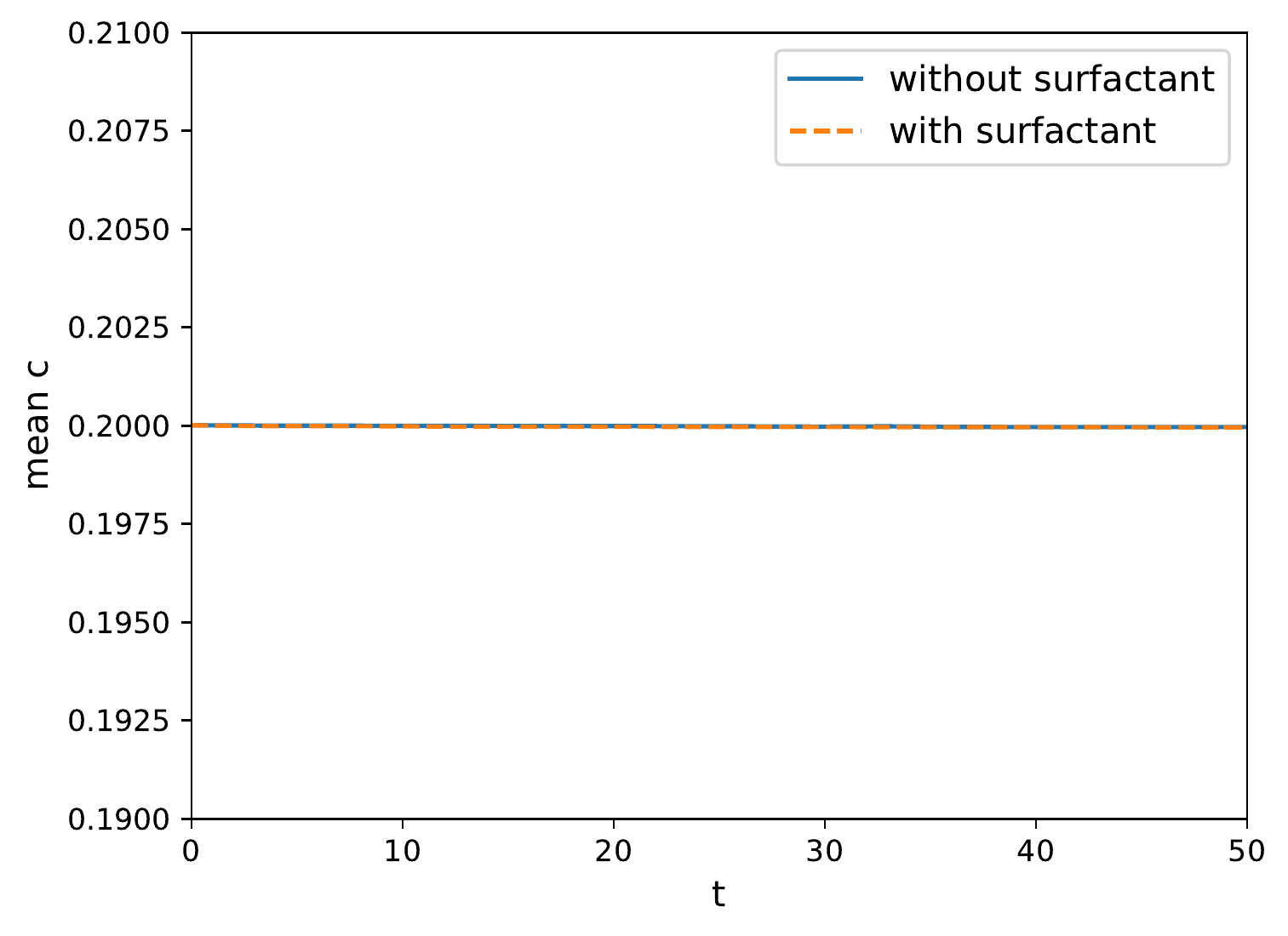}}
\subfigure[Mean $s$]{\includegraphics[width=0.32\textwidth]{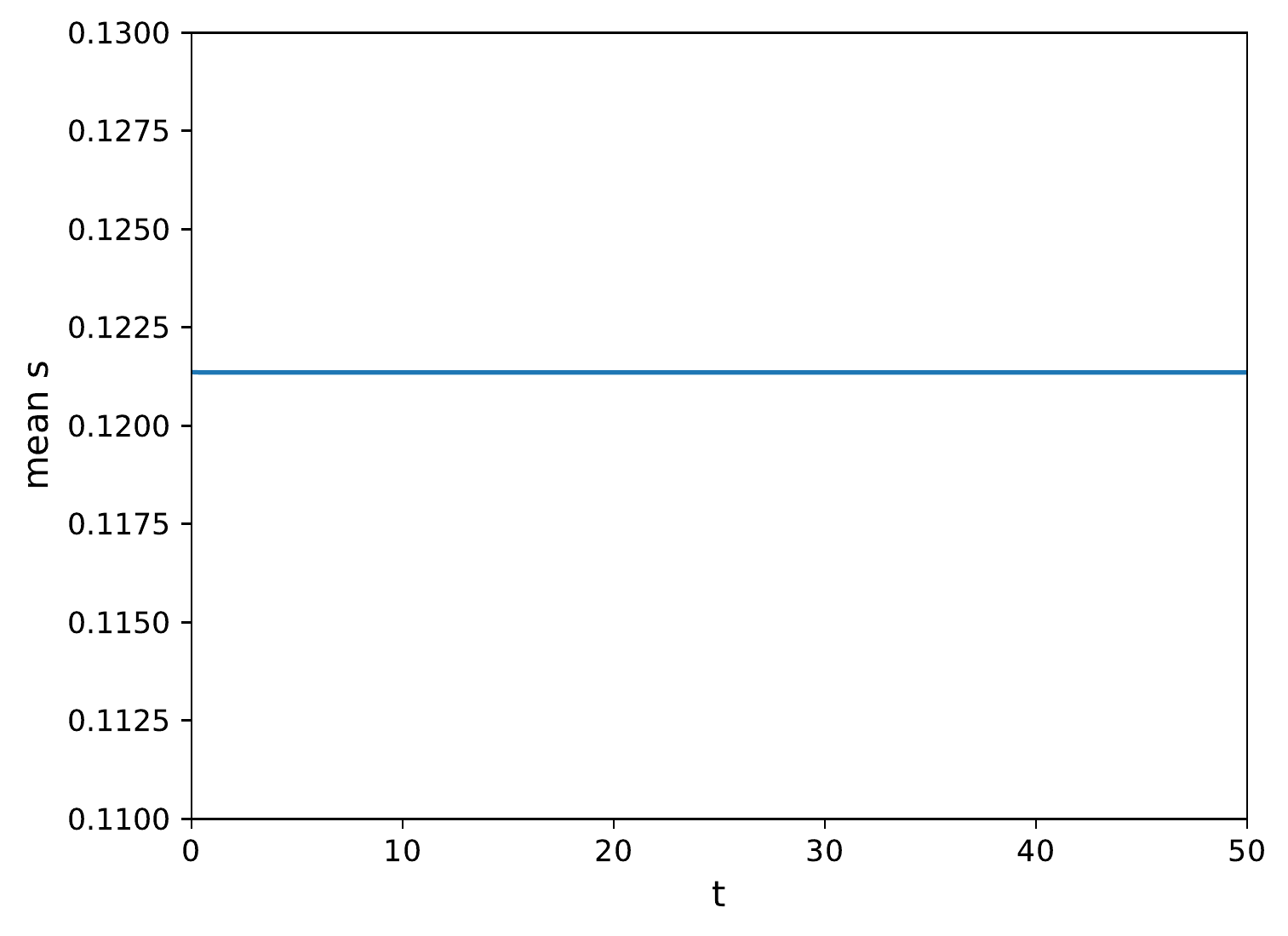}}
\caption{The preservation of mean $c$ and $s$ in the spinodal-drop problem. The curves for mean $c$ overlap with and without surfactant}
\label{fig:spinodal_mean}
\end{center}
\end{figure}

\begin{figure}[htbp]
\begin{center}
\subfigure[Free energy]{\includegraphics[width=0.32\textwidth]{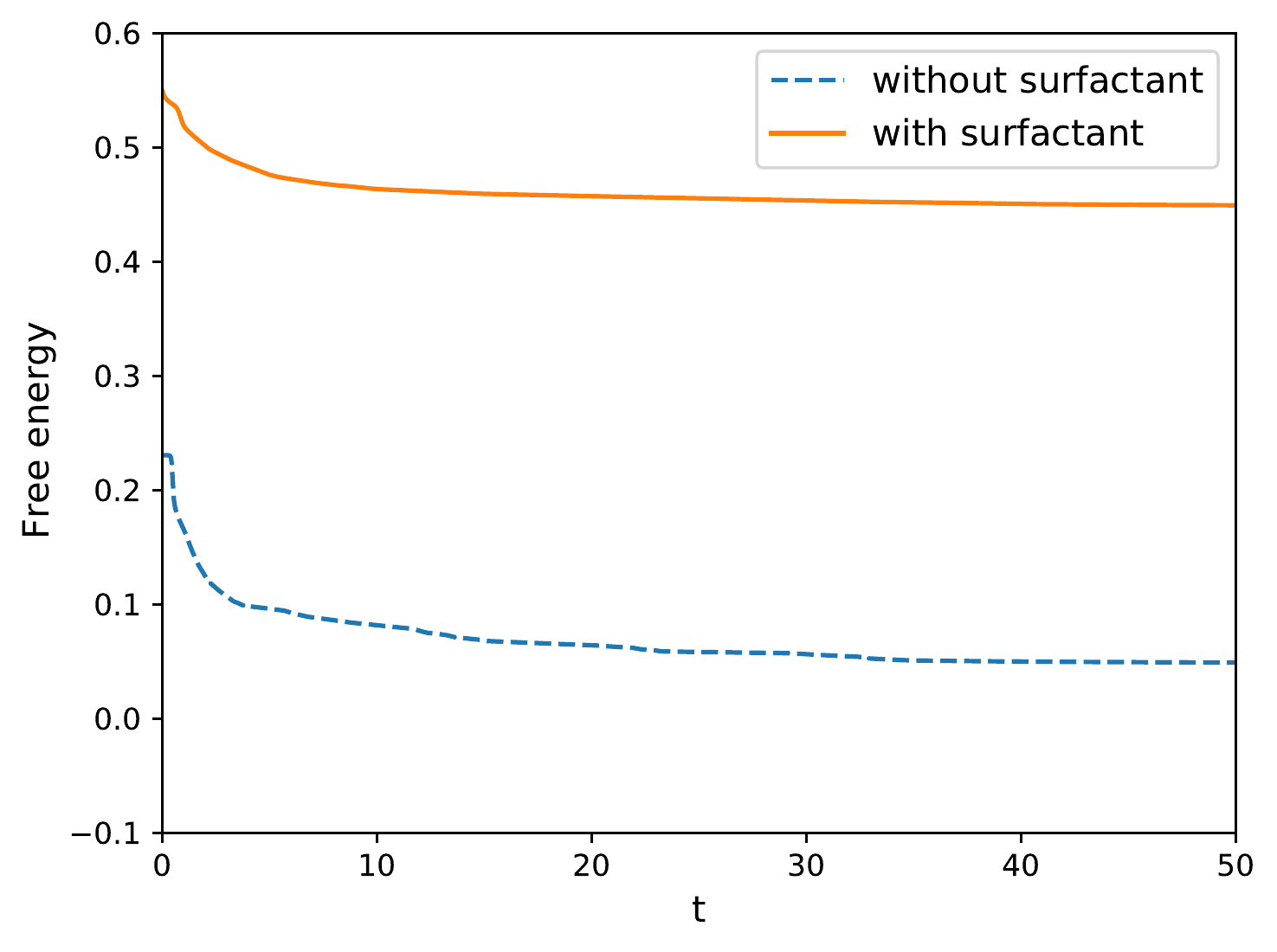}}
\subfigure[Free energy components in the presence of a surfactant]{\includegraphics[width=0.32\textwidth]{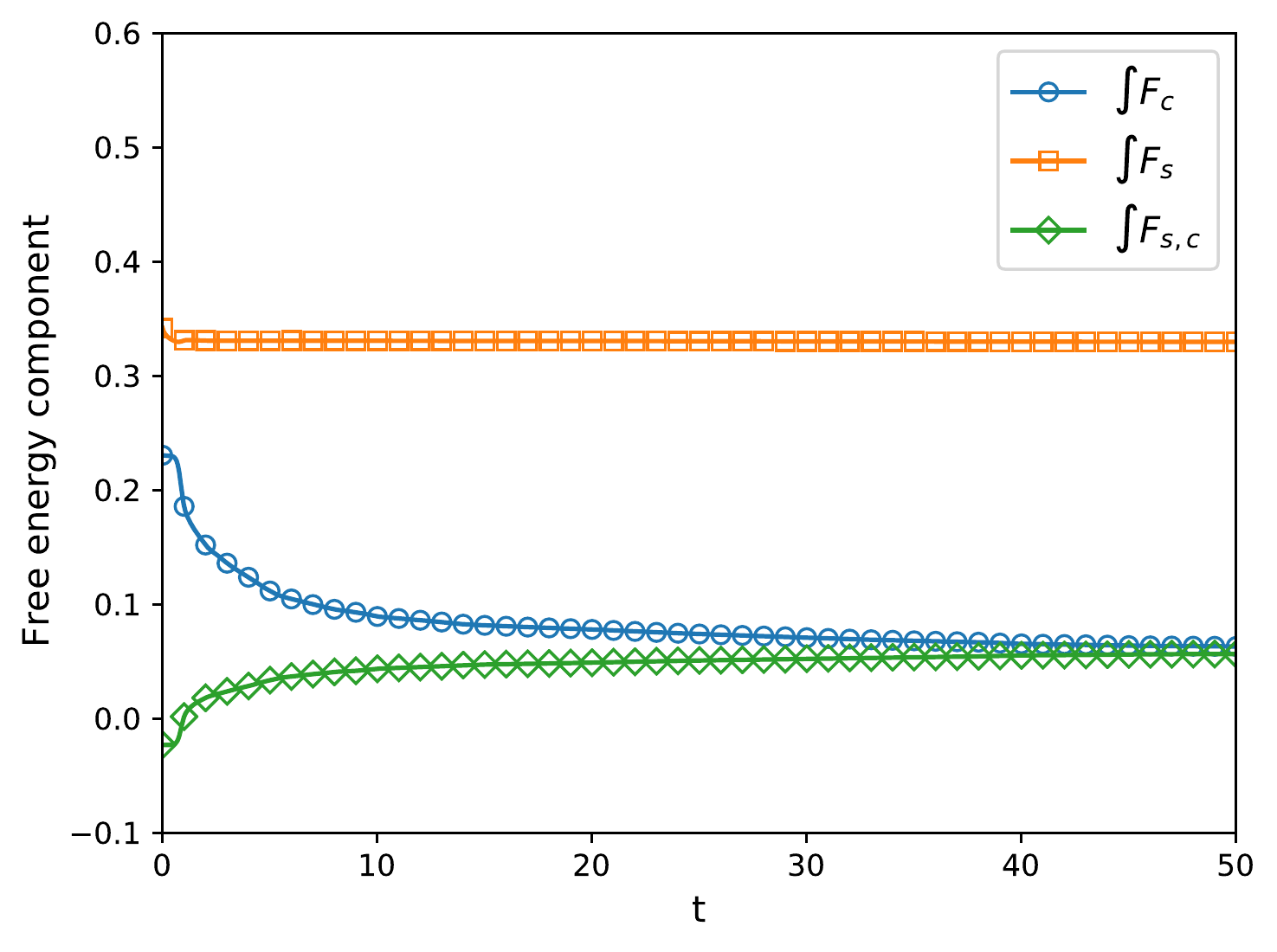}}
\caption{Evolution of Helmholtz free energy \eqref{eqn:helmenergy} for the spinodal-drop problem. (a) total free energy, (b) components of free energy}
\label{fig:spinodal_energy}
\end{center}
\end{figure}

\subsection{Flow through a cylinder}
We now consider the system dynamics in the presence of an underlying velocity field. The domain is the cylinder 
\[
\{(x,y,z) : \ \sqrt{(x-0.5)^2 + (y-0.5)^2} <0.25, \ z \in (0,1) \},
\]
which is discretized using cubic elements with edge length equal to $0.01$. Inflow and outflow boundary conditions are imposed at $z=0$ and $z=1$ respectively. The initial profile of the order parameter is given by $c^0(x,y,z) = \tanh((0.2-z)/\sqrt{2} \Cn))$ and the entire domain is initially filled with surfactant of concentration $s^0 = 0.01$.
The various parameters are set as  $\Pe_c = 100$, $\Pe_s = 100$ and $\alpha_2 = 1$. To study the effects of the interfacial adsorption (controlled by $\alpha_3$) and free surfactant penalization (controlled by $\alpha_4$), we choose $\alpha_3 \in \{0.5,1\}$ and $\alpha_4 \in \{0.5,1\}$. The simulation is run till time $t=0.7$. The velocity field is taken to be the steady state Poiseuille flow, which is depicted along a vertical cross-section through the cylinder axis in Figure \ref{fig:cylinder_v_c}(a). 

\begin{figure}[htbp]
\begin{center}
\subfigure[velocity field]{\includegraphics[width=0.22\textwidth]{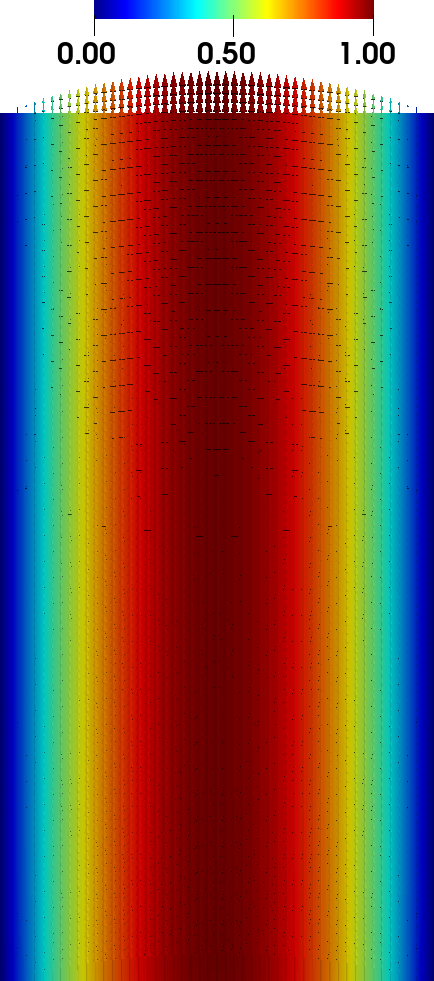}}
\subfigure[$c$ profile at final time]{\includegraphics[width=0.22\textwidth]{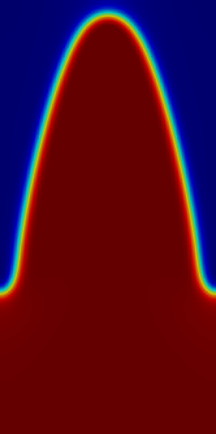}}
\caption{Steady state velocity field and order parameter profile at final time for flow through a cylinder. The plots are shown along a vertical cross-section through the cylinder axis}
\label{fig:cylinder_v_c}
\end{center}
\end{figure}

The evolution of the order parameter is indistinguishable for the various parameter combinations considered in this experiment, with the final profile shown in Figure \ref{fig:cylinder_v_c}(b). Note that the initial planar interface develops into a protruded interface due to the underlying velocity field. The extent of the protrusion can vary depending on the choice of $\Pe_c$, as has been observed in \cite{LIU20}. The final surfactant profiles for various combinations of $\alpha_3$ and $\alpha_4$ are shown in Figure \ref{fig:cylinder_s}. We observe that increasing $\alpha_4$ forces a larger amount of surfactant to move to the interface, as compared to increasing $\alpha_3$. This can be seen more clearly in Figure \ref{fig:cylinder_s_line} where we plot of the surfactant concentration along the cylinder axis at final time. 

\begin{figure}[htbp]
\begin{center}
\subfigure[$\alpha_3 = 0.5$, $\alpha_4 = 0.5$]{\includegraphics[width=0.22\textwidth]{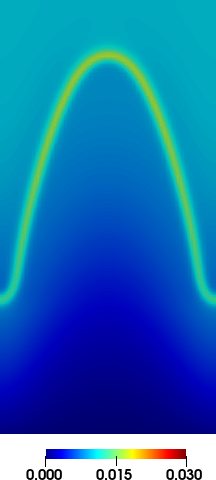}}
\subfigure[$\alpha_3 = 0.5$, $\alpha_4 = 1$]{\includegraphics[width=0.22\textwidth]{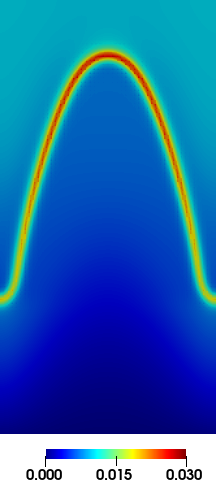}}
\subfigure[$\alpha_3 = 1$, $\alpha_4 = 0.5$]{\includegraphics[width=0.22\textwidth]{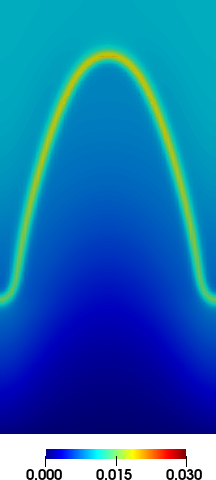}}
\subfigure[$\alpha_3 = 1$, $\alpha_4 = 1$]{\includegraphics[width=0.22\textwidth]{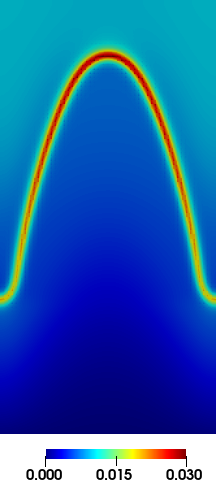}}
\caption{Surfactant profile at final time for flow through a cylinder for several values of $\alpha_3, \alpha_4$. The plots are shown along a vertical cross-section through the cylinder axis}
\label{fig:cylinder_s}
\end{center}
\end{figure}

\begin{figure}[htbp]
\begin{center}
\includegraphics[width=0.45\textwidth]{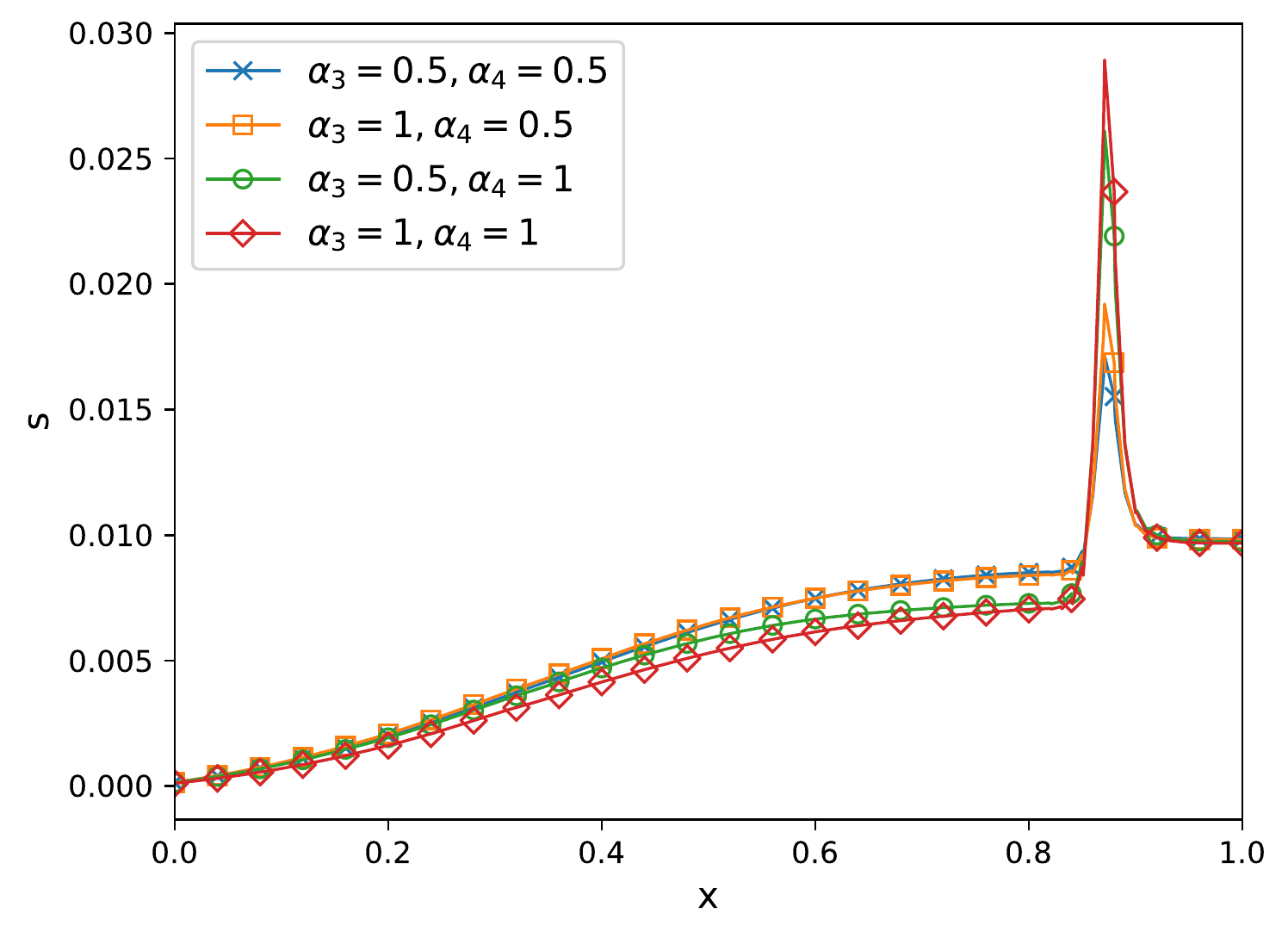}
\caption{Surfactant concentration at final time along the axis of the cylinder, i.e., $x=0.5,\ y=0.5,\ z \in [0,1]$}
\label{fig:cylinder_s_line}
\end{center}
\end{figure}

\subsection{Flow of a droplet through a sinusoidal pipe}
This experiment is designed to demonstrate the benefit of introducing a surfactant into a two-phase flow. We consider a sinusoidal pore space described by
\begin{equation*}
\{(x,y,z) : \sqrt{(y-0.5)^2 + (z-0.5)^2} < r(x), \ x \in (0,1)\},
\end{equation*}
where
\begin{equation*}
r(x) = \begin{cases}
\frac{(r_\text{pipe} - r_\text{throat})}{2}\cos\left(\frac{8(5x-1) \pi}{3} \right) + \frac{r_\text{pipe} +r_\text{throat}}{2} \\
\qquad \text{if} \ x \in (0.2,0.8)\\
r_{pipe}  \qquad  \text{if} \ x_0 \in (0,0.2] \cup [0.8,1).
\end{cases}
\end{equation*}
Here the radius of the pipe is $r_\text{pipe} = 0.1$ and the radius of each throat is $r_\text{throat} = 0.015$. The shape of the pipe is shown in Figure \ref{fig:sinepipe}, where the domain is discretized using cubic cells with edge length equal to $5 \times 10^{-3}$. The velocity field is obtained by solving the incompressible Navier-Stokes equation to steady state in this domain, by considering the inflow boundary condition 
\[
\vel_{in} = 0.02 \left(1- \left(\frac{y-0.5}{r_\text{pipe}} \right)^2 - \left(\frac{z-0.5}{r_\text{pipe}} \right)^2 \right)
\]
at $x=0$ and setting open/Neumann boundary conditions at the outlet $x=1$. The magnitude of the steady-state velocity is shown in Figure \ref{fig:sinepipe_v_c_init}(a). Note that the velocity magnitude is the largest at the throats of the pipe. The initial profile for $c$ is given by
\begin{align*}
c^0(x,y,z) &= -\tanh\left(\frac{0.04 - d}{\sqrt{2} \Cn}\right), \\
d &= \sqrt{(x-0.35)^2 + (y-0.5)^2 + (z-0.5)^2},
\end{align*}
and is shown in Figure \ref{fig:sinepipe_v_c_init}(b). This describes the scenario of a residual oil drop (blue phase) trapped inside a pore. The remaining parameters are chosen as $\Pe_c = 100$, $\Pe_s = 100$, $\alpha_2 = 1$, $\alpha_3 = 1$ and $\alpha_4 = 1$.

\begin{figure}[htbp]
\begin{center}
\includegraphics[width=0.42\textwidth]{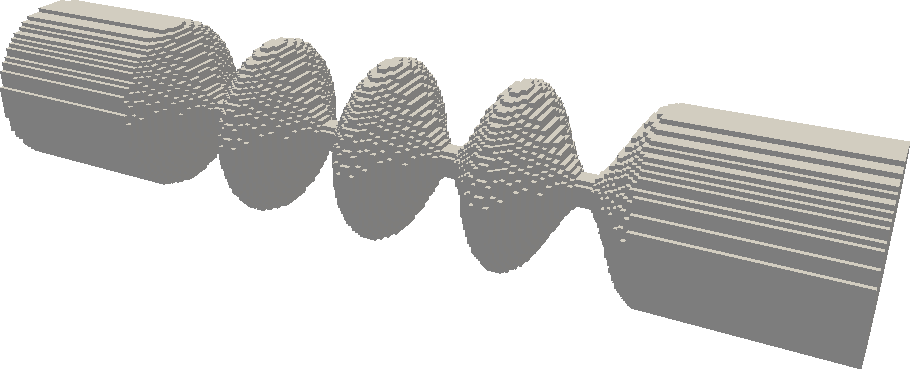}
\caption{Shape of sinusoidal pipe, where the flow domain is shown in gray}
\label{fig:sinepipe}
\end{center}
\end{figure}

In the absence of any surfactant, the oil drop passes through the throat into the second cavity, as shown in Figure \ref{fig:sinepipe_s0_c}. However, it is unable to detach itself from the walls of the pipe and gets stuck at time $t \approx1.3$ (also see Figure \ref{fig:sinepipe_vol}(b)). We restart the simulation and introduce a constant initial surfactant throughout the entire pipe, $s^0 = 0.01$. As expected, the surfactant moves from the bulk phase and adsorbs to the interface of the drop (see Figure \ref{fig:sinepipe_s0p01_s}). However,  the drop once again gets stuck to the wall at $t \approx 1.3$, as can be seen in Figure \ref{fig:sinepipe_s0p01_c}. We repeat the experiment again, but this time taking a larger amount for the initial surfactant,  $s^0 = 0.05$. For this case, the amount of surfactant adsorbed on the interface seems to be sufficient to push the drop of oil in the next cavity of the pipe, as shown in Figure \ref{fig:sinepipe_s0p05_c} (also see Figure \ref{fig:sinepipe_vol}(c)). Thus, one can hope to push out trapped oil from cavities by introducing a sufficient amount of surfactant. This also motivates the use of surfactants in enhanced oil recovery from oil reservoirs. 

In addition, we note that the radius of the oil drop is considerably reduced when it is successfully pushed into neighbouring cavity (see \ref{fig:sinepipe_s0p05_c}(e)-(f)). Eventually the drop completely diffuses into the domain. As shown in Figure \ref{fig:sinepipe_s0p05_s}, the surfactant collapses into a drop once the diffusive interface of $c$ disappears. This spontaneous shrinkage of a drop is known to occur with the Cahn-Hilliard system, when the radius of the drop is smaller than a critical radius \cite{YUE07}, which for the current problem is given by
\[
r_c = \left( \frac{2^{1/6}}{3 \pi} V \Cn \right)^{1/4} \approx 0.0923,
\] 
where $V \approx 0.1218$ is the volume of each pore, while $\Cn = 5 \times 10^{-3}$. Note that the radius of the initial drop is $0.04$, which is much smaller than the critical radius.

\begin{figure}[htbp]
\begin{center}
\subfigure[velocity magnitude]{\includegraphics[width=0.48\textwidth]{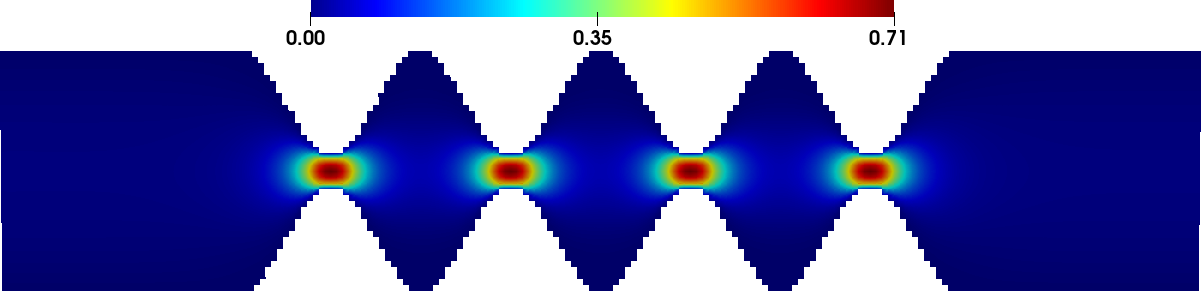}}
\subfigure[$c$ profile at initial time]{\includegraphics[width=0.48\textwidth]{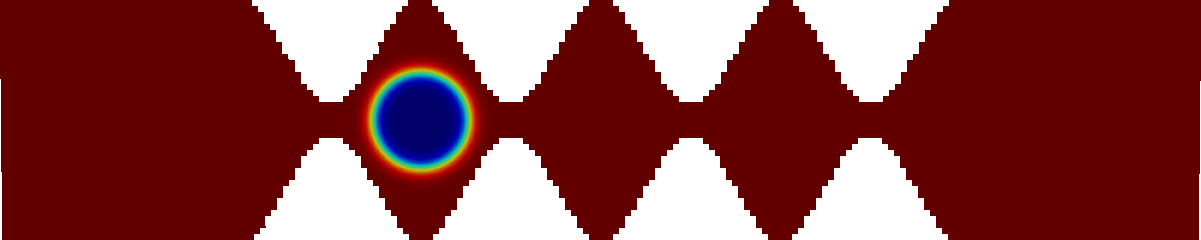}}
\caption{Steady state velocity field and the initial order parameter profile for flow through a sinusoidal pipe. The plots are shown along a vertical cross-section through the axis of the pipe}
\label{fig:sinepipe_v_c_init}
\end{center}
\end{figure}

\begin{figure}[htbp]
\begin{center}
\subfigure[$t=0.2$]{\includegraphics[width=0.48\textwidth]{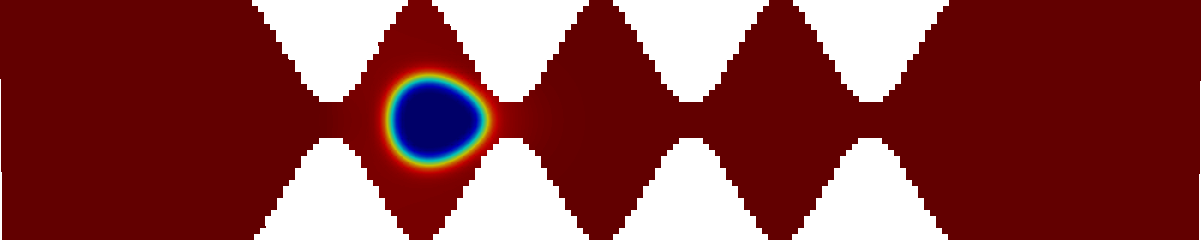}}
\subfigure[$t=0.5$]{\includegraphics[width=0.48\textwidth]{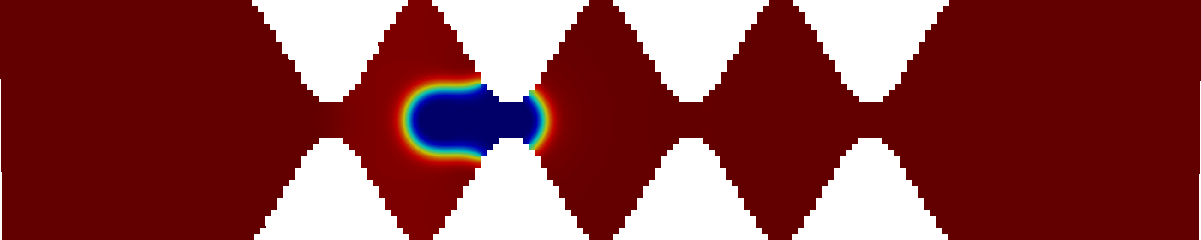}}
\subfigure[$t=1.0$]{\includegraphics[width=0.48\textwidth]{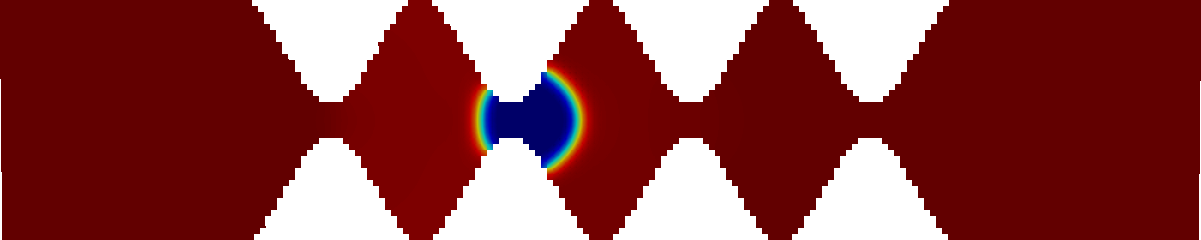}}
\subfigure[$t=1.3$]{\includegraphics[width=0.48\textwidth]{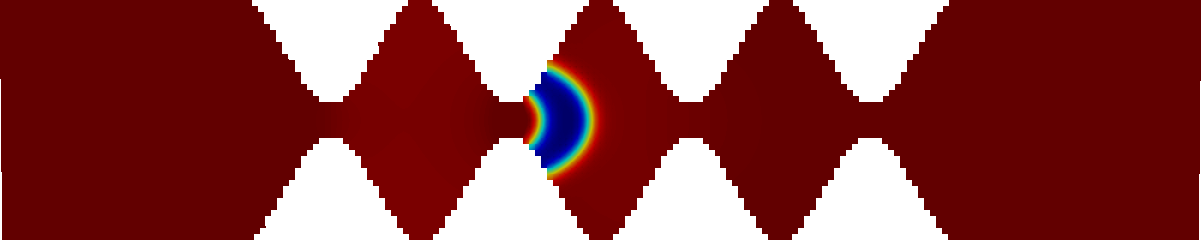}}
\caption{Snapshots of flow of a trapped oil drop in the absence of any surfactant}
\label{fig:sinepipe_s0_c}
\end{center}
\end{figure}

\begin{figure}[htbp]
\begin{center}
\subfigure[$t=0.2$]{\includegraphics[width=0.48\textwidth]{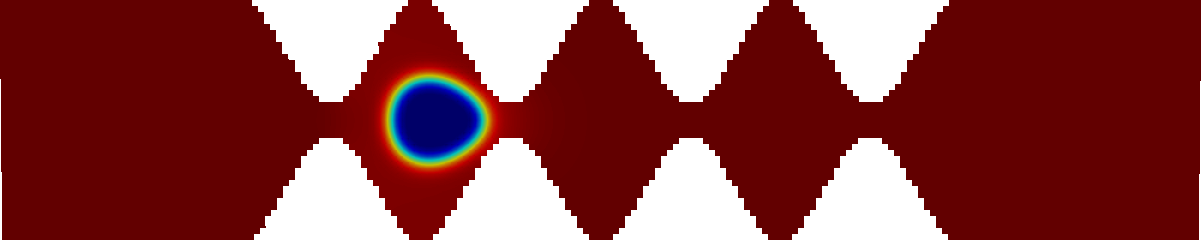}}
\subfigure[$t=0.5$]{\includegraphics[width=0.48\textwidth]{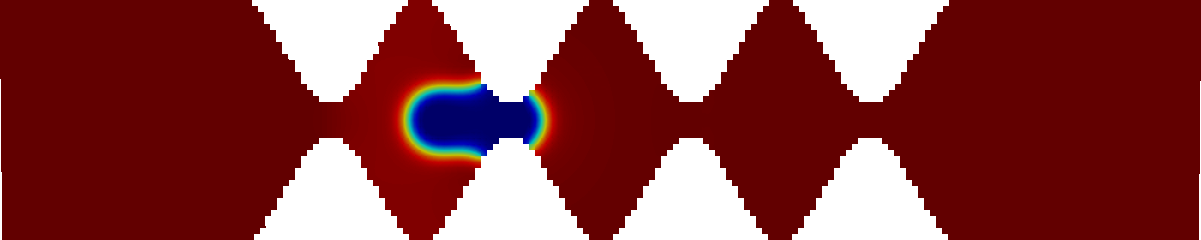}}
\subfigure[$t=1.0$]{\includegraphics[width=0.48\textwidth]{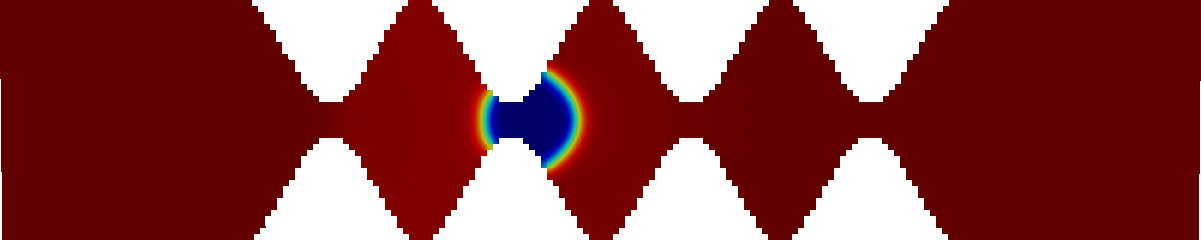}}
\subfigure[$t=1.3$]{\includegraphics[width=0.48\textwidth]{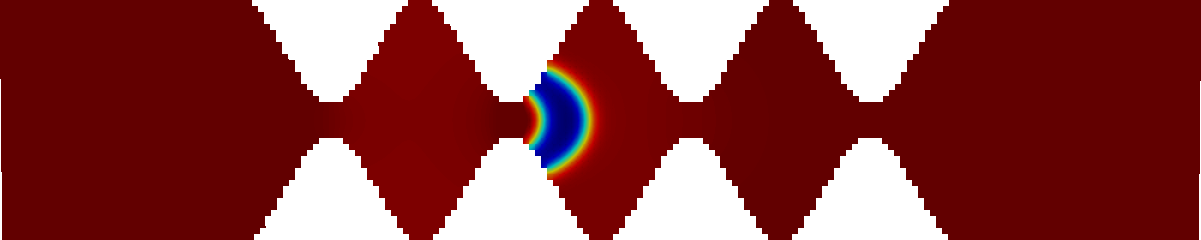}}
\caption{Snapshots of flow of a trapped oil drop with an initial constant surfactant $s^0 = 0.01$}
\label{fig:sinepipe_s0p01_c}
\end{center}
\end{figure}

\begin{figure}[htbp]
\begin{center}
\subfigure[$t=0.2$]{\includegraphics[width=0.48\textwidth]{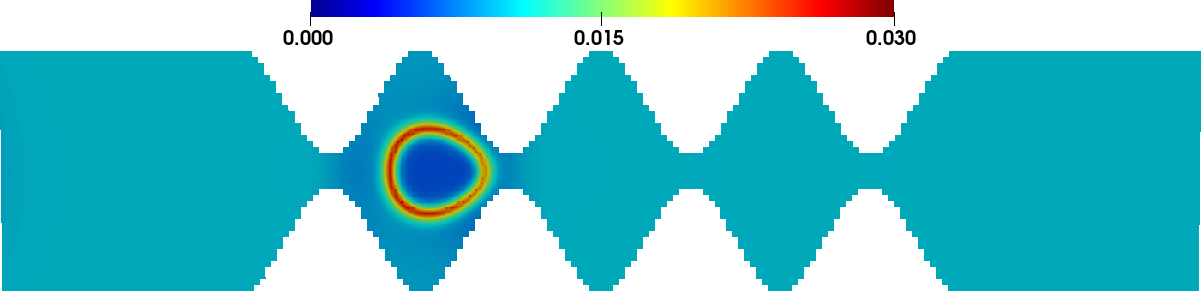}}
\subfigure[$t=0.5$]{\includegraphics[width=0.48\textwidth]{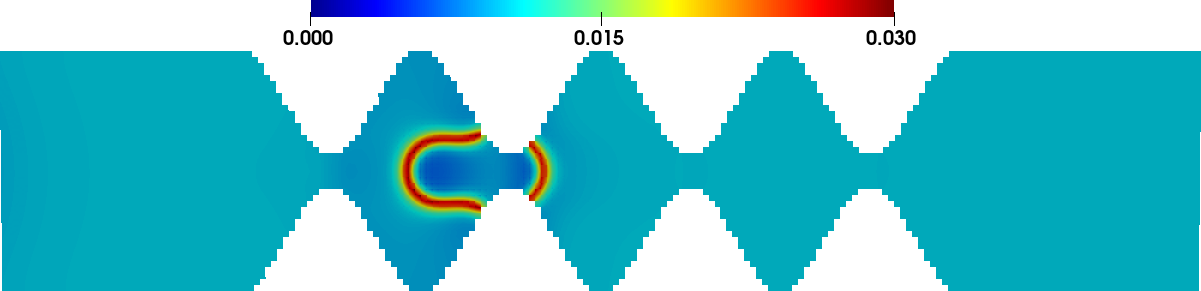}}
\subfigure[$t=1.0$]{\includegraphics[width=0.48\textwidth]{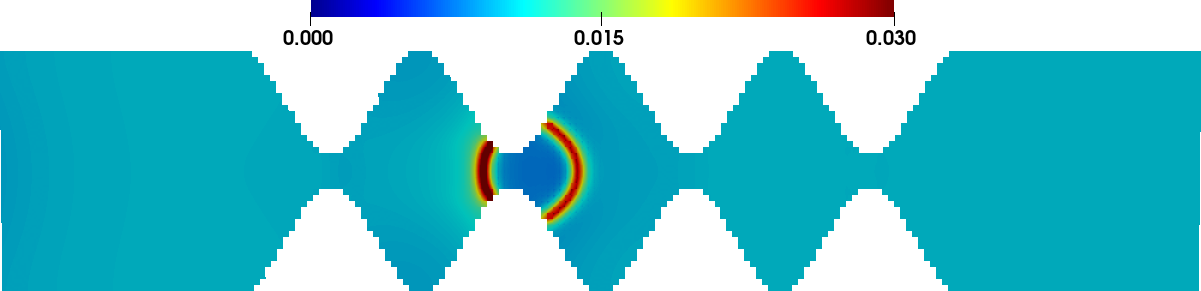}}
\subfigure[$t=1.3$]{\includegraphics[width=0.48\textwidth]{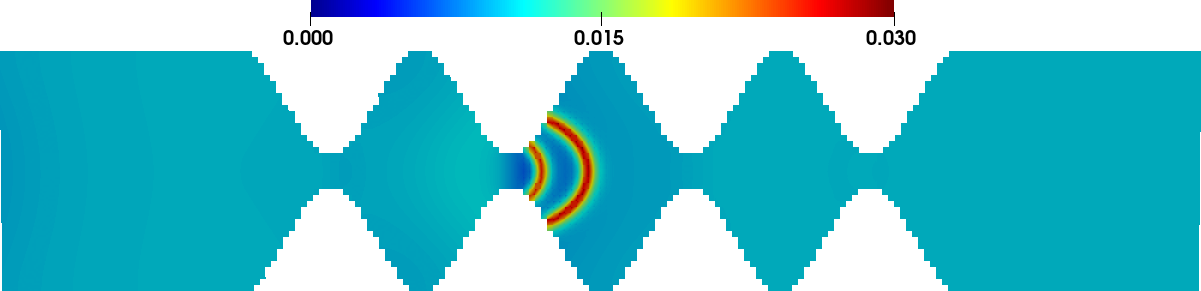}}
\caption{Surfactant dynamics in a flow of a trapped oil drop with an initial constant surfactant $s^0 = 0.01$}
\label{fig:sinepipe_s0p01_s}
\end{center}
\end{figure}

\begin{figure}[htbp]
\begin{center}
\subfigure[$t=0.2$]{\includegraphics[width=0.48\textwidth]{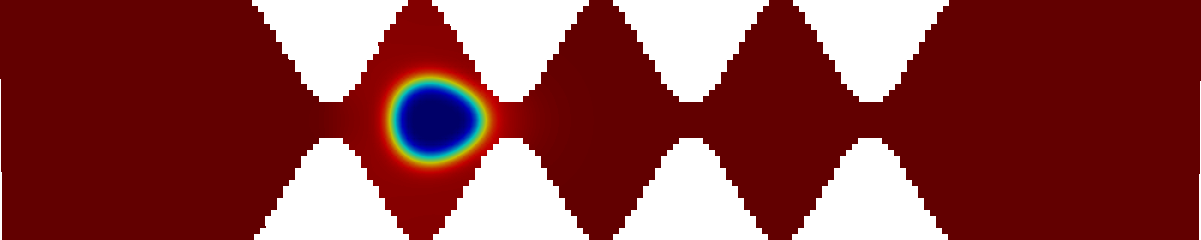}}
\subfigure[$t=0.5$]{\includegraphics[width=0.48\textwidth]{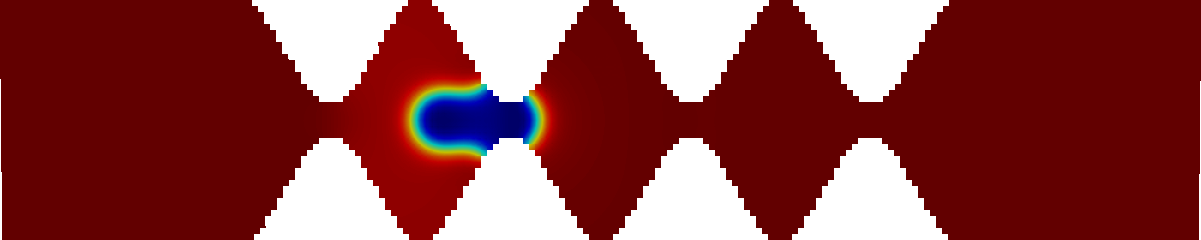}}
\subfigure[$t=1.0$]{\includegraphics[width=0.48\textwidth]{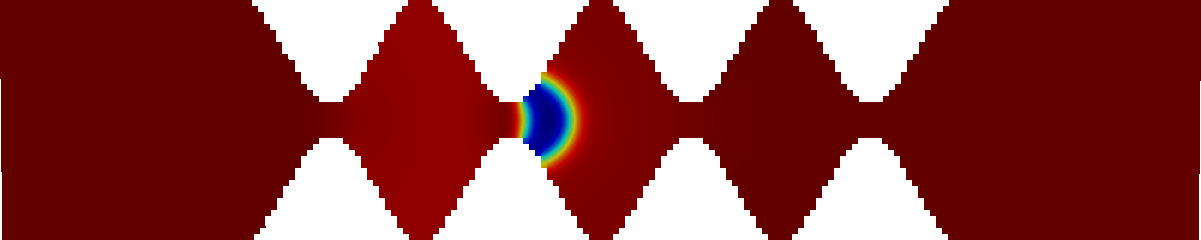}}
\subfigure[$t=1.3$]{\includegraphics[width=0.48\textwidth]{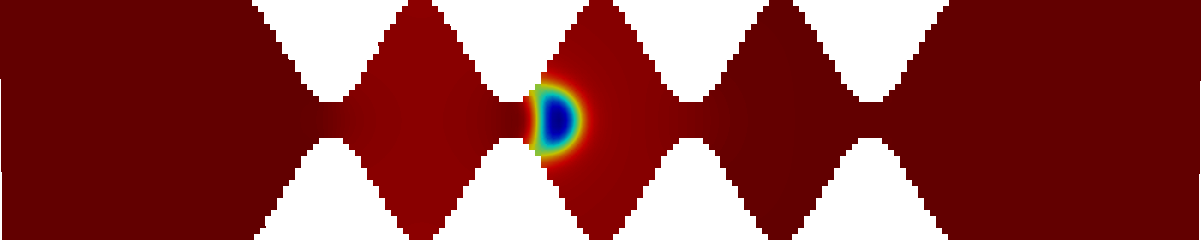}}
\subfigure[$t=1.4$]{\includegraphics[width=0.48\textwidth]{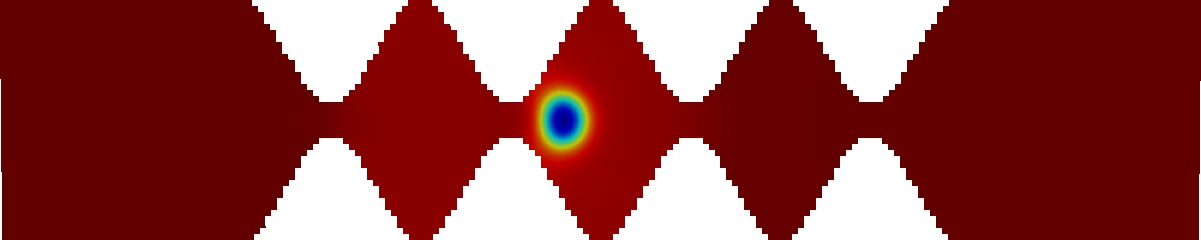}}
\subfigure[$t=1.5$]{\includegraphics[width=0.48\textwidth]{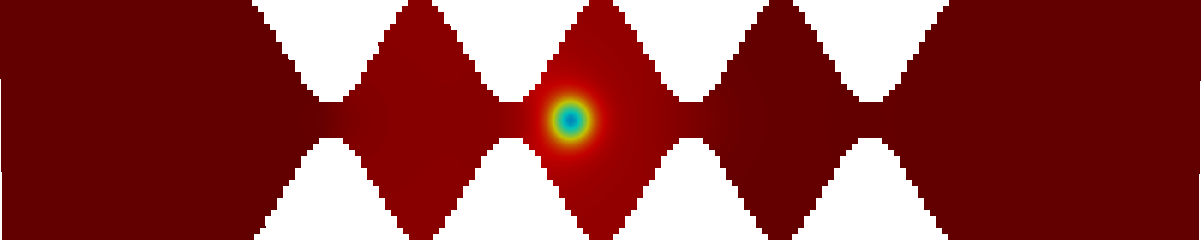}}
\caption{Snapshots of flow of a trapped oil drop with an initial constant surfactant $s^0 = 0.05$}
\label{fig:sinepipe_s0p05_c}
\end{center}
\end{figure}

\begin{figure}[htbp]
\begin{center}
\subfigure[$t=0.2$]{\includegraphics[width=0.48\textwidth]{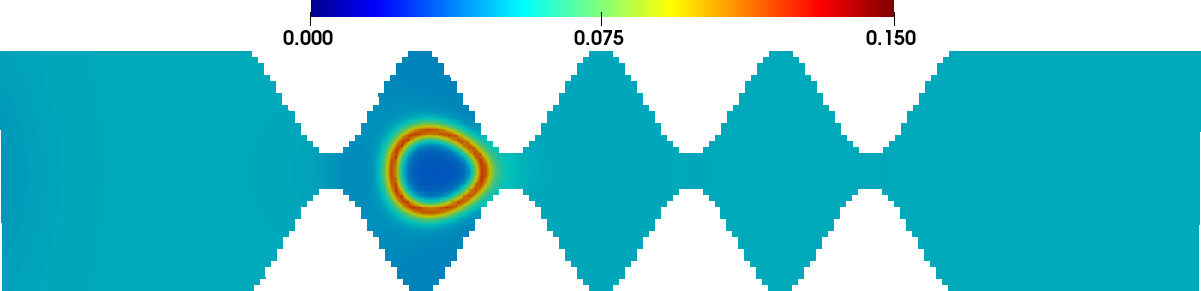}}
\subfigure[$t=0.5$]{\includegraphics[width=0.48\textwidth]{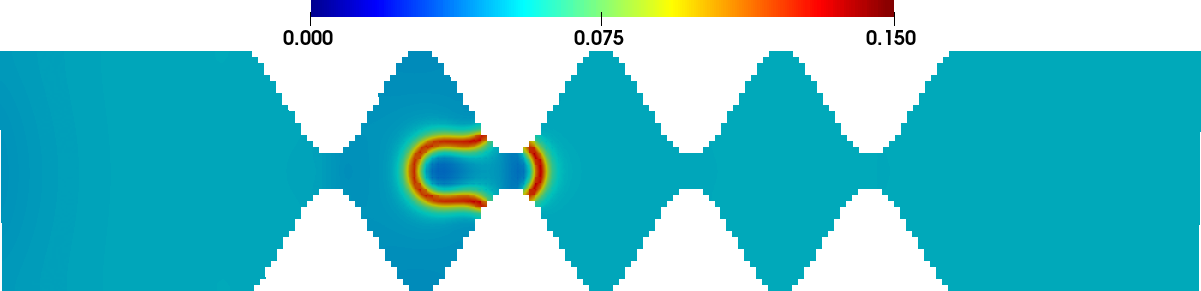}}
\subfigure[$t=1.0$]{\includegraphics[width=0.48\textwidth]{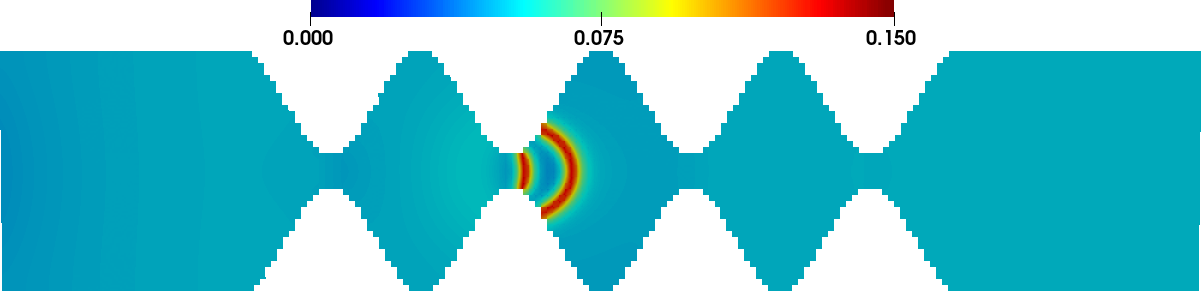}}
\subfigure[$t=1.3$]{\includegraphics[width=0.48\textwidth]{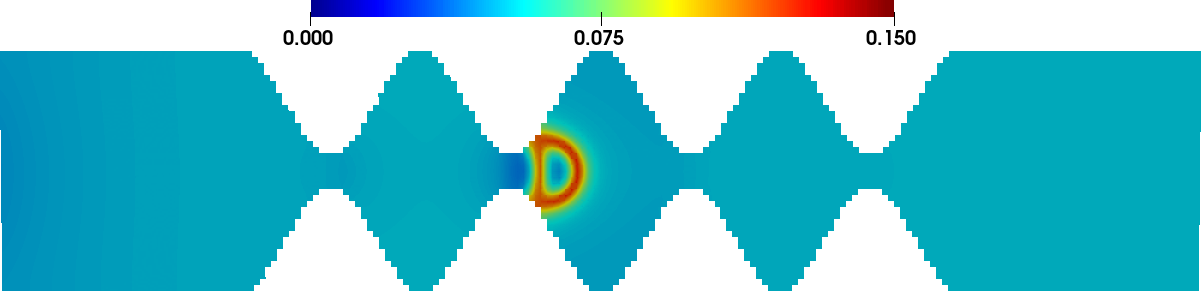}}
\subfigure[$t=1.4$]{\includegraphics[width=0.48\textwidth]{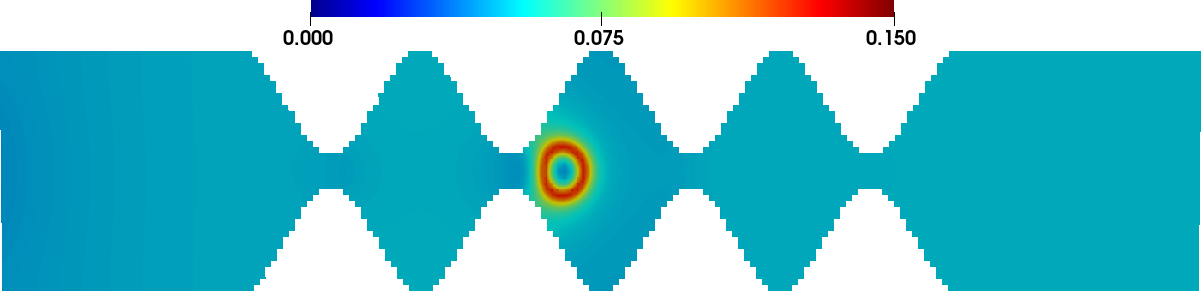}}
\subfigure[$t=1.5$]{\includegraphics[width=0.48\textwidth]{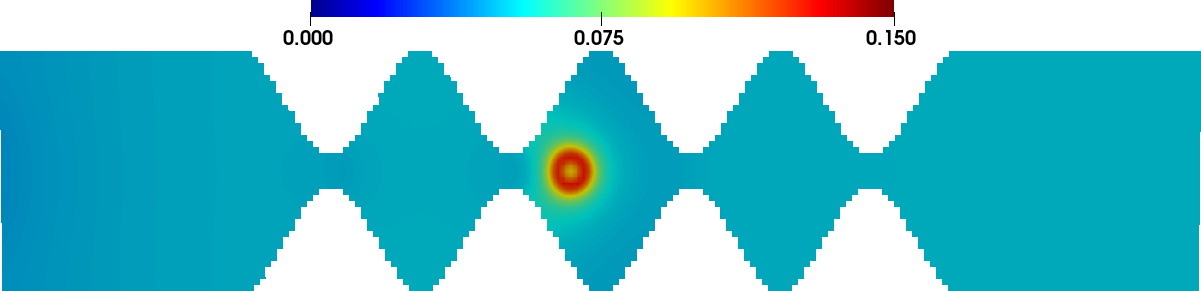}}
\caption{Surfactant dynamics in a flow of a trapped oil drop with an initial constant surfactant $s^0 = 0.05$}
\label{fig:sinepipe_s0p05_s}
\end{center}
\end{figure}

\begin{figure}[htbp]
\begin{center}
\subfigure[$t=0.3$, without surfactant]{\includegraphics[width=0.32\textwidth]{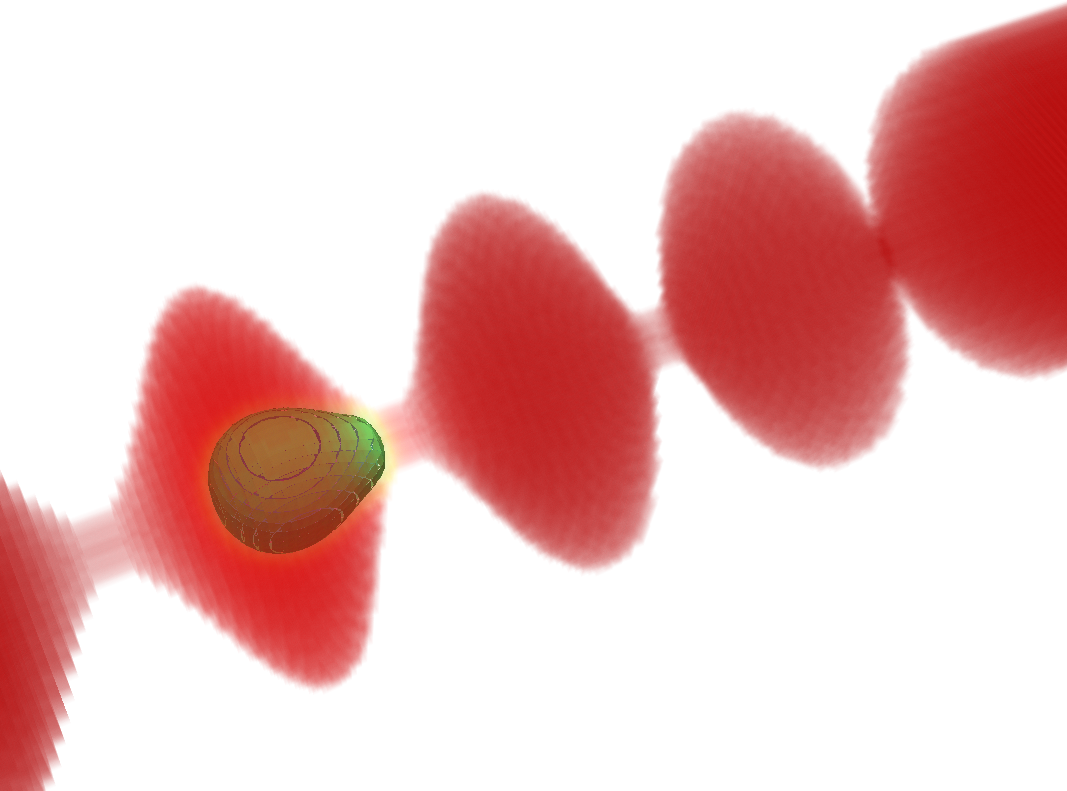}}
\subfigure[$t=1.3$, without surfactant]{\includegraphics[width=0.32\textwidth]{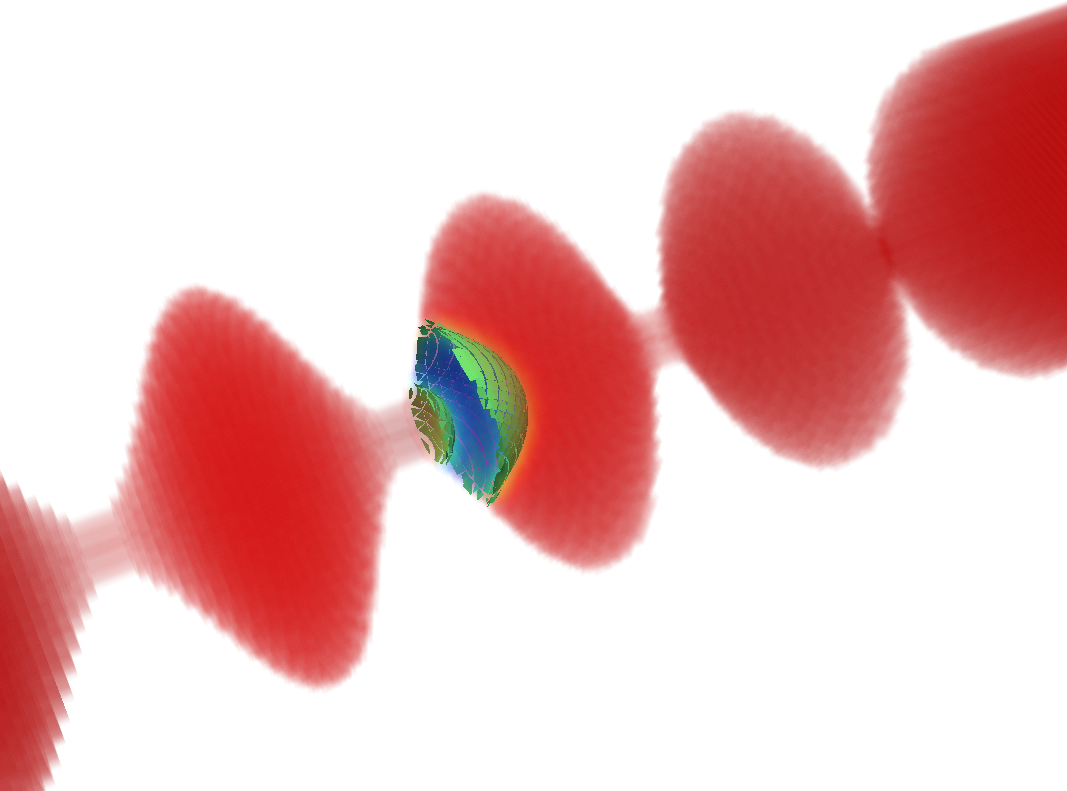}}
\subfigure[$t=1.3$, with surfactant ($s^0=0.05$)]{\includegraphics[width=0.32\textwidth]{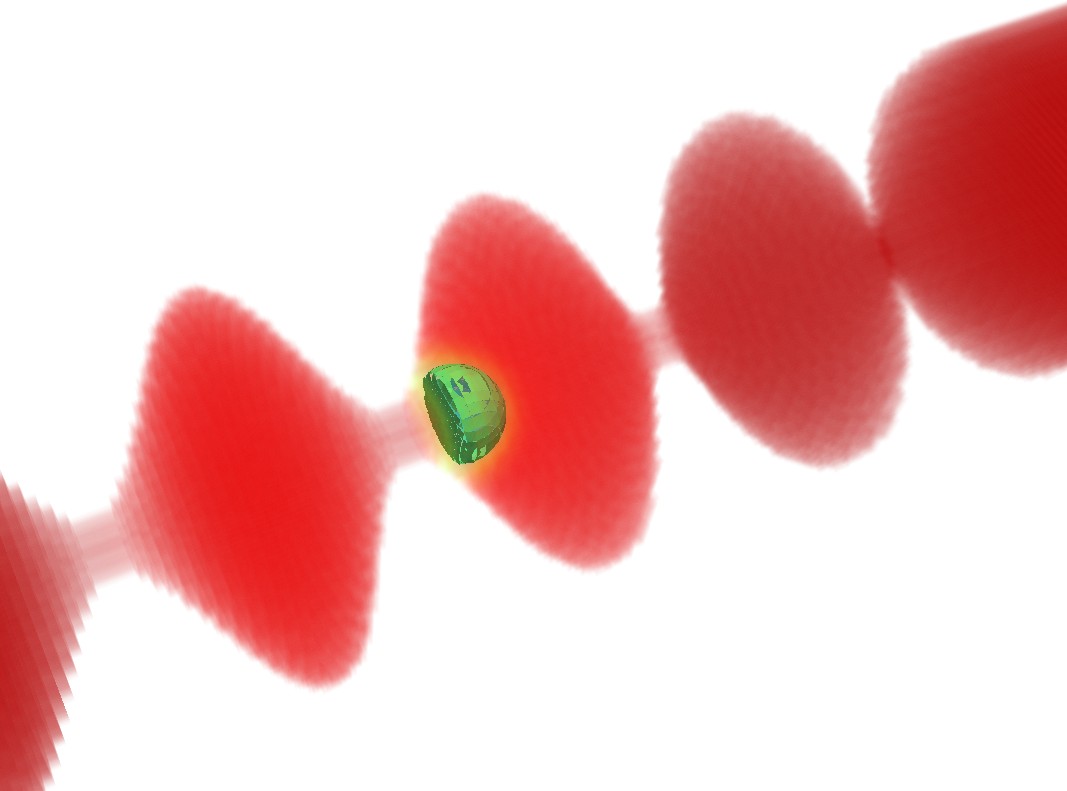}}
\caption{Three-dimensional depiction of flow of a trapped oil drop through a sinusoidal pipe. The red phase represents water, the blue phase represents oil and the green surface denotes the diffuse interface i.e., $c=0$}
\label{fig:sinepipe_vol}
\end{center}
\end{figure}

\subsection{Flow through Berea sandstone}
Finally, we simulate a realistic flow through a porous media, where the domain is generated by micro-CT scans of a Berea sandstone sample \cite{ANDRA2013}. In Figure \ref{fig:berea_full}, we show the rock sample embedded in the domain $(0,1)^3$, which is discretized with a mesh size $h=1/160$. The domain inflow is set at $x=0$, the outflow at $x=1$, while all remaining boundaries are set as solid walls. In order to induce a stable flow field in the pore space, we have attached buffers at the inflow and outflow faces, each having a width of 16 cells. The underlying velocity field is obtained by solving the incompressible Navier-Stokes to time $t=1$, which is shown in Figure~\ref{fig:berea_soln}(a).

The surfactant-order parameter system is solved with parameters $\Pe_c = 100$, $\Pe_s = 100$, $\alpha_2 = 1$, $\alpha_3 = 1$ and $\alpha_4 = 1$. A uniform time-step of $\tau=5 \times 10^{-3}$ is used to march in time. The pore space is initially saturated with one of the phases ($c=-1$), while the second phase is injected through the inflow. We also consider the pore space to be initially saturated with a minimal surfactant concentration of $s^0=10^{-3}$, while a constant stream of surfactant with $s=0.2$ is injected into the domain along with the second phase. The profiles of $c$ and $s$ at time $t=1$ are shown in Figure \ref{fig:berea_soln}(b)-(c). We observe that the surfactant concentration is much higher at the diffusive interface ($c=0$) in any local neighbourhood of the domain. To visualize the dynamics in the interior of the domain, we consider the solution on 2D slices in the direction of the flow. The surfactant adsorbs to the interface on each of these slices, as can be seen in Figure \ref{fig:berea_slice}.

\begin{figure}[htbp]
\begin{center}
\includegraphics[width=0.42\textwidth]{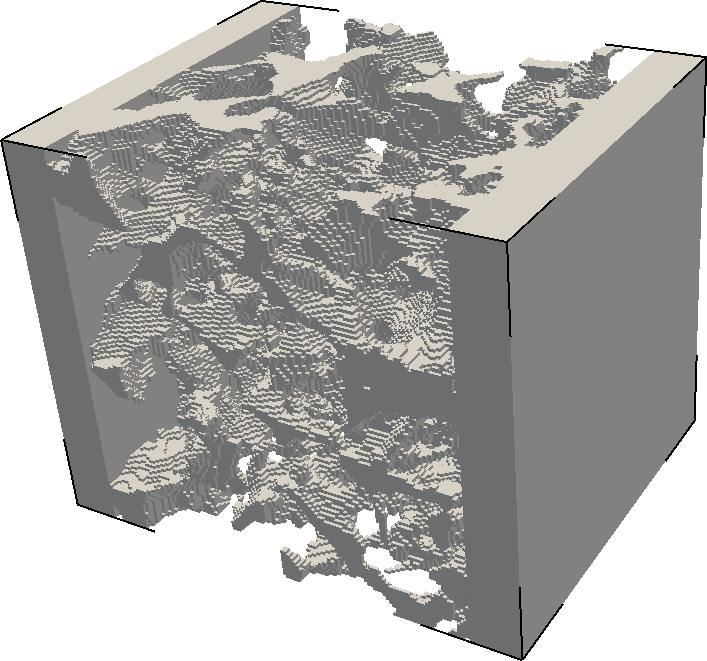}
\caption{Berea sandstone structure with the flow domain shown in gray}
\label{fig:berea_full}
\end{center}
\end{figure}

\begin{figure}[htbp]
\begin{center}
\subfigure[Velocity field]{\includegraphics[width=0.35\textwidth]{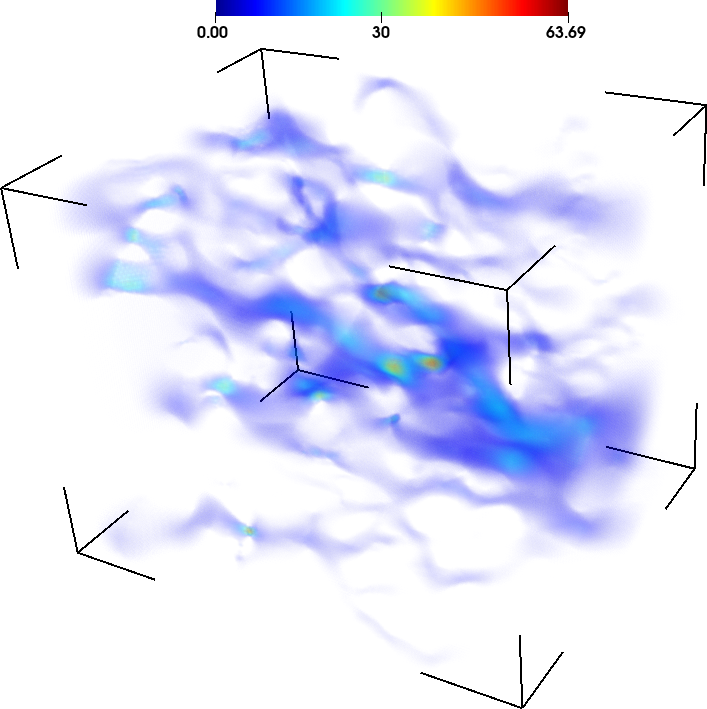}}
\subfigure[$c$]{\includegraphics[width=0.35\textwidth]{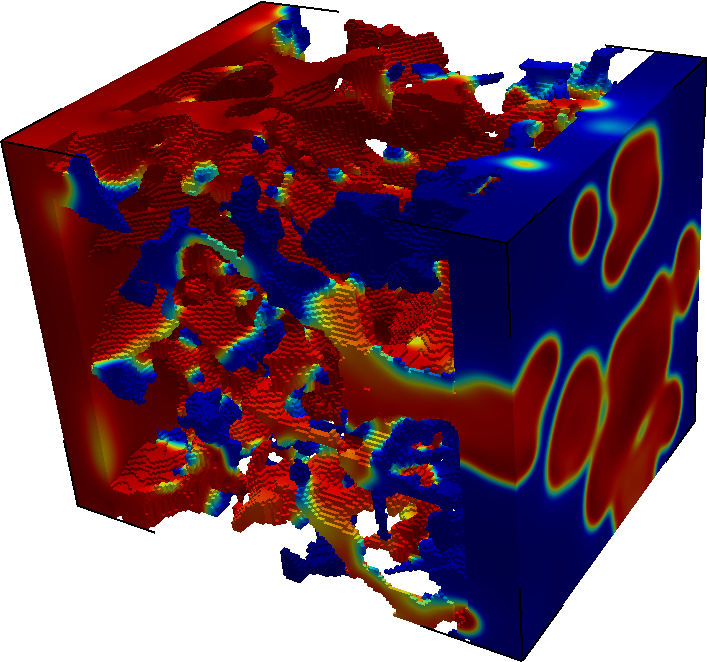}}
\subfigure[$s$]{\includegraphics[width=0.35\textwidth]{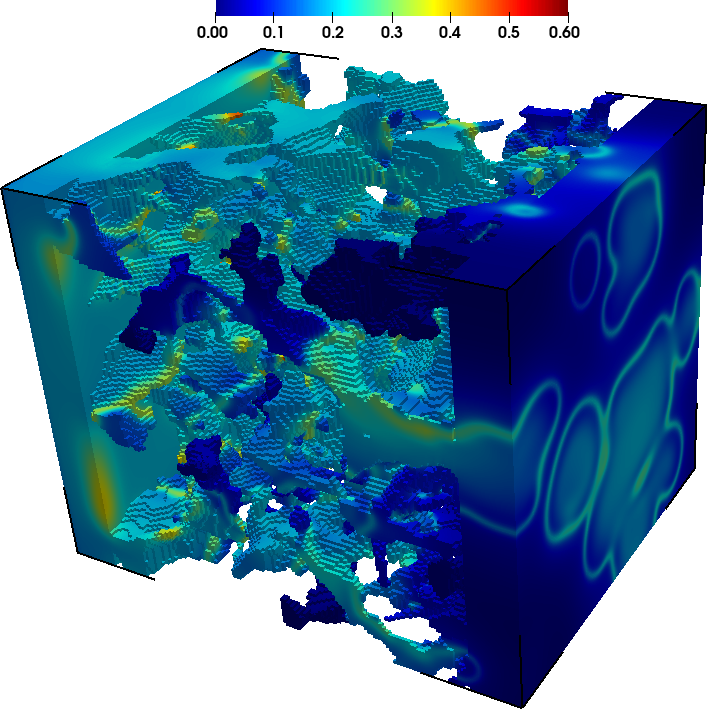}}
\caption{Solution of a three-component system through Berea sandstone at time $t=1$}
\label{fig:berea_soln}
\end{center}
\end{figure}

\begin{figure}[htbp]
\begin{center}
\subfigure[$c$ through $z=0.2$]{\includegraphics[width=0.35\textwidth]{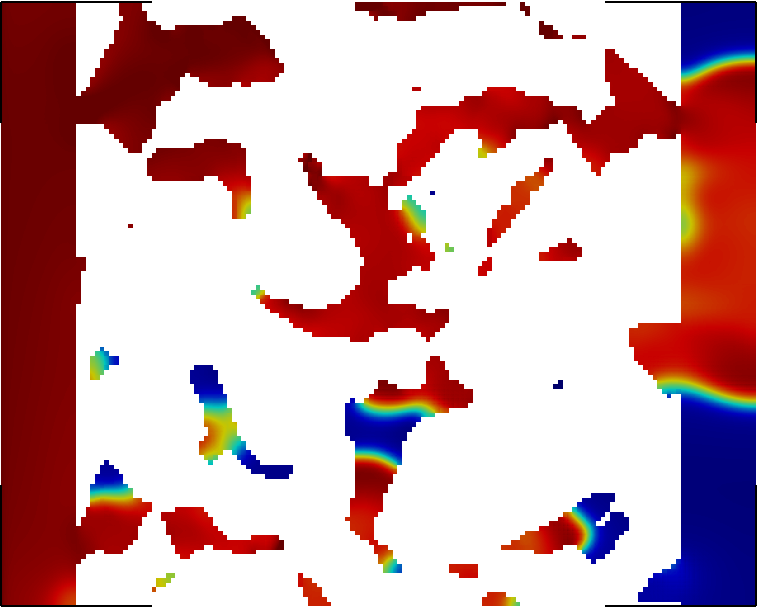}}
\subfigure[$s$ through $z=0.2$]{\includegraphics[width=0.35\textwidth]{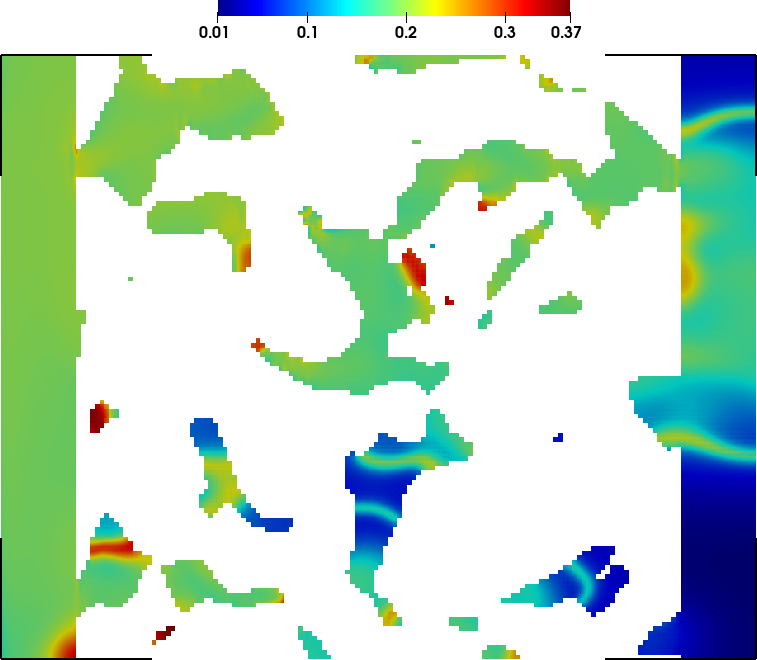}}\\
\subfigure[$c$ through $z=0.5$]{\includegraphics[width=0.35\textwidth]{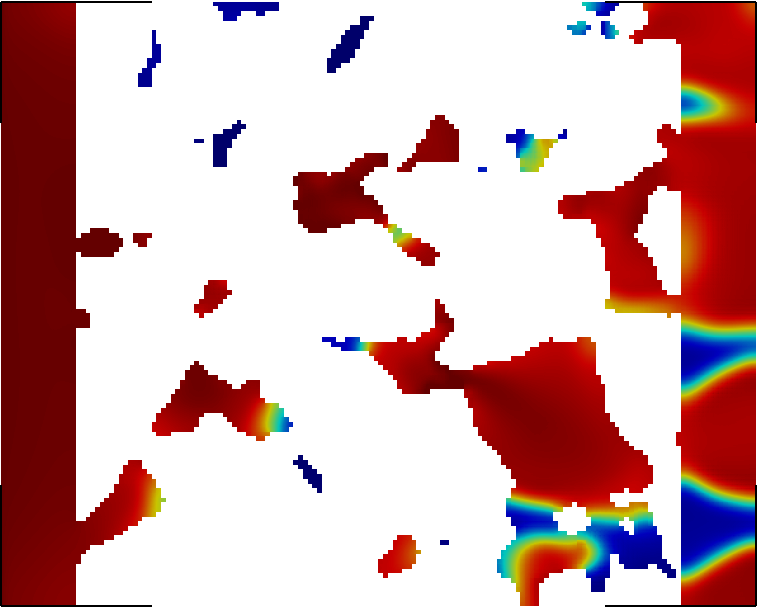}}
\subfigure[$s$ through $z=0.5$]{\includegraphics[width=0.35\textwidth]{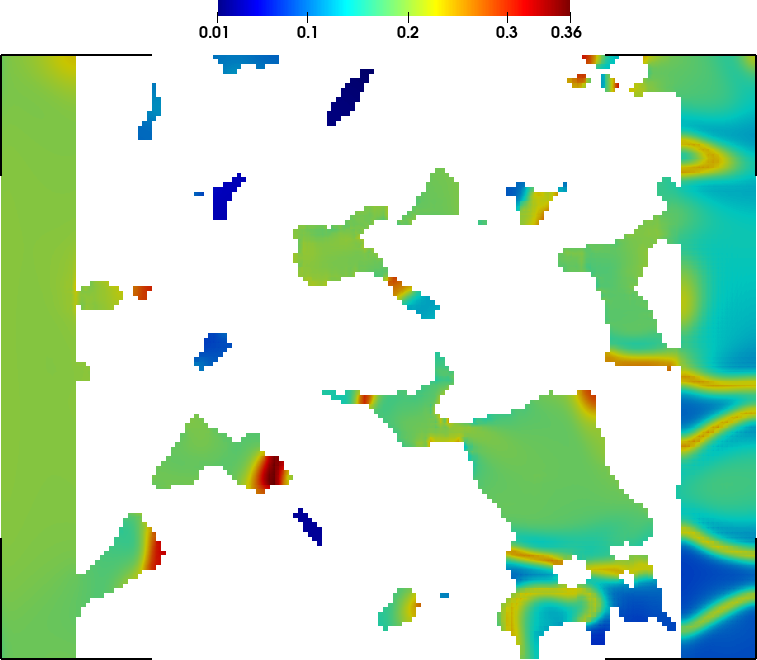}}\\
\subfigure[$c$ through $z=0.8$]{\includegraphics[width=0.35\textwidth]{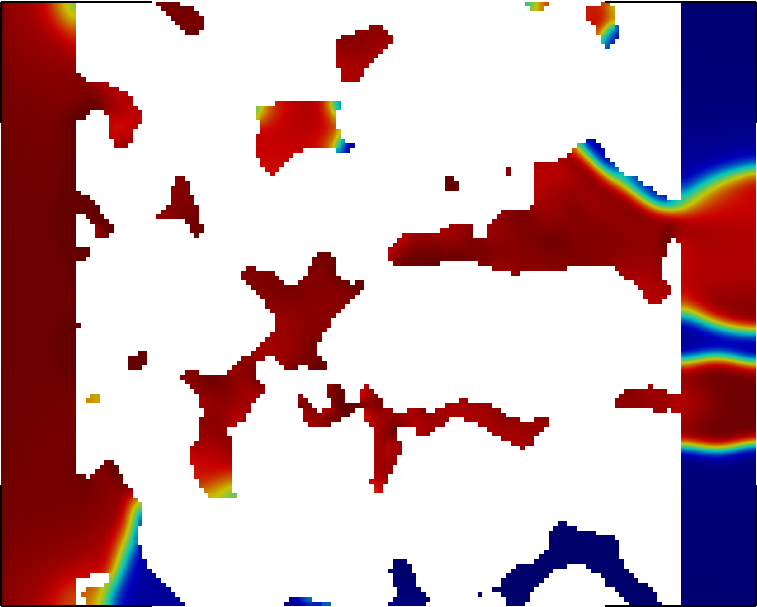}}
\subfigure[$s$ through $z=0.8$]{\includegraphics[width=0.35\textwidth]{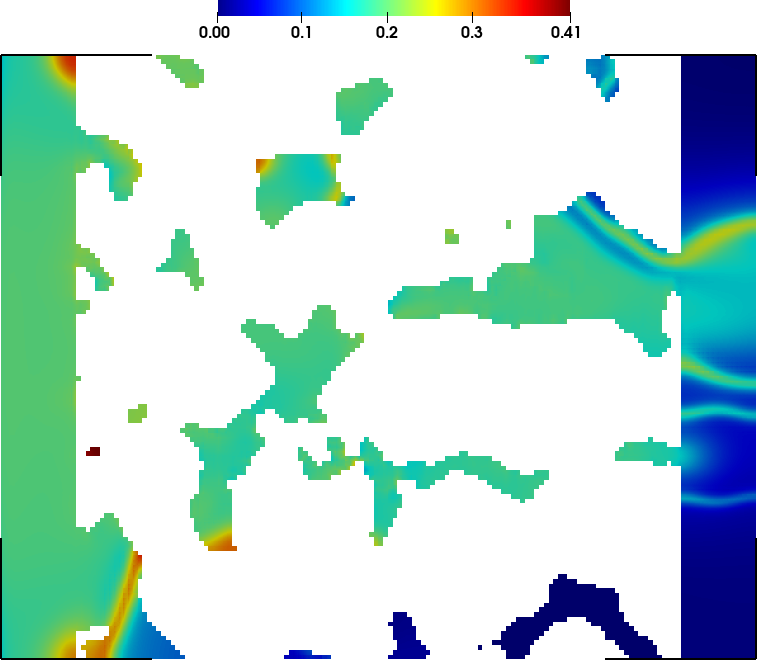}}
\caption{Order parameter and surfactant on 2D slices of the Berea sandstone at $t=1$}
\label{fig:berea_slice}
\end{center}
\end{figure}

%\end{document}

%% !TEX root = adv_surch_scheme.tex
%\documentclass[adv_surch_scheme.tex]{subfiles}
%
%\begin{document}

\section{Conclusion}
This work formulates a numerical scheme for the discretization of a phase-field model of a system of two immiscible phases and a soluble surfactant.  The method is based on the discontinuous Galerkin method in space and a concave-convex splitting
in time.  Numerical results demonstrate that the scheme recovers the Langmuir adsorption isotherms, while exhibiting desirable physical properties, such as the decay of total Helmholtz free-energy. The simulation results show that the surfactant's concentration is (locally) higher at the interface between the two phases. The results also show the impact of the surfactant in facilitating the motion of trapped bubbles in pores. Finally, the proposed scheme is used to simulate flow through a Berea rock sample, thereby establishing its utility in effectively solving realistic problems.

This work demonstrates that an IPDG scheme can be used to solve the two-phase flow problem in the presence of a surfactant, which is known to be quite challenging. The DG formulation allows us to achieve arbitrary order of accuracy in space, even in complicated porous domains. While the time-discretization used in this paper is only first-order accurate, higher-order time marching strategies that ensure the decay of total free energy needs to be explored. Furthermore, the model considered in the work assumes that the underlying velocity field is not affected by the phase-surfactant dynamics. To capture more realistic dynamics would require a two-way coupling between the phase-surfactant model and the underlying incompressible flow equations is required. This will be investigated in future work, along with the construction of high-order time marching strategies that ensure the decay of total energy.
%
%\end{document}

\section*{Acknowledgement}
The authors thank Dr. Steffen Berg for useful discussions on surfactant models.
Ray and Riviere acknowledge funding from a Shell-Rice collaboration. Riviere is also partially funded by NSF-DMS 1913291.

\bibliographystyle{abbrv}
\bibliography{ref}

\end{document}